\documentclass[12pt]{article}

\usepackage{amsfonts}
\usepackage{eurosym}
\usepackage{graphicx}
\usepackage{booktabs}
\usepackage{float}
\usepackage{amssymb}
\usepackage{comment}
\usepackage{amsmath}
\usepackage{breqn}
\usepackage{mathrsfs}
\usepackage{booktabs}
\usepackage{multirow}
\usepackage{enumerate}
\usepackage{color}
\usepackage[authoryear]{natbib}
\usepackage{hyperref}
\usepackage{rotating}
\usepackage{accents}
\usepackage{bm}
\usepackage{bbm}
\usepackage{mathtools}
\usepackage{tikz}
\usetikzlibrary{patterns}
\usepackage{subcaption}
\usepackage{enumerate}
\usepackage{breqn}
\usepackage{amsthm}
\usepackage{xcolor}
\usepackage{setspace}
\usepackage[multiple]{footmisc}

\definecolor{airforceblue}{rgb}{0.46, 0.54, 0.76}
\definecolor{battleshipgrey}{rgb}{0.52, 0.52, 0.51}
\definecolor{charcoal}{rgb}{0.21, 0.27, 0.31}
\hypersetup{colorlinks=true, linkcolor=airforceblue, citecolor=airforceblue, urlcolor=black}

\newtheorem{lemma}{{\bf \sc Lemma}}
\newtheorem{example}{{\bf \sc Example}}

\newtheorem{proposition}{{\bf \sc Proposition}}

\newtheorem{fact}{{\bf \sc Fact}}

\def\E{{\mathrm{E}}}

\DeclareRobustCommand{\ubar}[1]{\underaccent{\bar}{#1}}

\newcommand\supp{{\mathrm{support}}}
\newcommand\pool{{\mathrm{pool}}}

\begin{document}


\title{The Role of Referrals in Immobility, Inequality, and Inefficiency in Labor Markets}
\author{Lukas Bolte, Nicole Immorlica and Matthew O. Jackson}
\date{
}
\maketitle

\begin{abstract}
We study the consequences of job markets' heavy reliance on referrals. Referrals lead to more opportunities for workers to be hired, which lead to better matches and increased productivity, but also disadvantage job-seekers with few or no connections to employed workers, increasing inequality.  Coupled with homophily, referrals also lead to immobility. We identify conditions under which distributing referrals more evenly reduces inequality and improves future productivity and mobility. We use the model to examine the short and long-run welfare impacts of policies such as affirmative action and algorithmic fairness.\footnote{We gratefully acknowledge support under NSF grants SES-1629446 and SES-2018554, as well as the Simons Institute for the Theory of Computing. We thank seminar participants from the AEA meetings, Cornell University, CREST, Econometric Society World Congress, European Meeting on Networks, Gamesnet, HSE Moscow, INET Networks Group, Lancaster University, NSF-Network Science in Economics Conference, University of Padova, U Penn, Penn State, Stockholm School of Economics, Toulouse School of Economics, Uppsala University, as well as Nageeb Ali, Francis Bloch, Ozan Candogan, Matt Elliott, and Moritz Meyer-ter-Vehn for helpful comments.}
\end{abstract}

\section{Introduction}

In 2012, Alan Krueger famously demonstrated the ``Great Gatsby Curve,'' which showed that countries with higher income inequality also tend to have higher immobility (i.e., a higher correlation between child's and parents' income). Thus, inequality within a generation is accompanied by immobility across generations. Understanding the common sources of inequality and immobility is essential to designing policies that deal not only with their symptoms but also their causes.\footnote{See \cite{corak2016,jackson2019,jackson2021} for additional discussion and references on inequality and its relationship to immobility.} Although much attention has been given to the role of wealth and family in the accumulation of human capital, less attention has been devoted to modeling the role of social capital---the connections that people have---in driving inequality and immobility. This is despite the fact that it has been shown that the role of people's social capital can be enormous.  For instance, \cite{chettyetal2022I} have recently shown that a specific form of social capital---economic connectedness---completely mediates the Gatsby Curve in US data. That is, once one accounts for the presence of friendships across income levels, the relationship between inequality and immobility becomes insignificant: the correlation operates via the presence/absence of cross-class friendships.\footnote{See Columns (5) and (6) in the second tier of Table~2 (the row labeled ``Income Inequality (Gini Coefficient)'').  Note also that the R$^2$ more than doubles.
See also \cite{laschever2013,sacerdote2001,chettyetal2016,beaman2012} for causal analyses.}

Here, we show how and why the absence of cross-group connections---homophily in social networks---coupled with the reliance of labor markets on referrals in hiring leads to a combination of inequality, economic immobility, and lost productivity. Our model simultaneously explains a  range of patterns seen in the data and also provides new insights into policies that overcome the negative impacts of homophily in referrals.

Our model is built on two key facts.  One is that many workers---of a variety of skill levels---are hired via referrals, and people applying without some connection are at a substantial disadvantage. 
The other is that these connections typically exhibit substantial homophily---the tendency of people to be friends with others with similar characteristics: ethnicity, gender, age, income level, religion, etc. 
Perhaps surprisingly, despite extensive studies of homophily,\footnote{See, e.g., \cite{mcphersonsc2001,currarinijp2009,jackson2019,jackson2021}} as well as the role of referrals in labor markets,\footnote{See, e.g., \cite{myerss1951,reess1970,granovetter1974,granovetter1995,montgomery1991,ioannidesd2004,topa2011,rubineauf2013,lesterrt2021,hederosskp2022}} 
there is little work combining homophily with referrals.  

Motivated by these facts, we introduce a dynamic model with overlapping generations of workers where current workers refer workers from the next generation to their employer.  Importantly, referrals exhibit homophily so that workers are more likely to refer workers of their own group (race, gender, etc.). Firms can either hire a worker referred to them or from the pool of unreferred and previously rejected workers. 

Referrals lower search frictions.   In particular, referrals provide an extra opportunity for workers to get hired and thus convey an advantage to referred workers both in terms of employment and wages.  Homophily coupled with the advantage bestowed by referrals helps us understand why inequality is so strongly tied with immobility: one group's current employment advantage then translates to the next generation via additional opportunities to be hired.\footnote{This is consistent with recent findings of the relationship between cross-class connections and economic mobility (e.g., \cite{chettyetal2022I,chettyetal2022II}).}  This also has
productivity implications, as people in one group may get multiple referrals while others in another group get none, which leads to fewer overall hiring opportunities and, thus, less efficient matching.  These are the essential components of our model.

We also model an extra feature of referrals, based on extensive empirical evidence.  In particular, in our model, referrals provide information about a potential worker's productivity, which is not available from a match in the pool of open applicants. Thus, referred workers are only hired if their value is better than the expected value from the open application process, and hence those hired via referrals are more productive, on average, than those hired through the open application process.\footnote{For evidence of such an effect see, e.g., \cite{rees1966,fernandezcm2000,brownst2016,pallaiss2016,lesterrt2021}).} Moreover, since homophily leads referrals to be concentrated among a subset of the population when one group starts out with an employment advantage, it follows that fewer workers with high productivities are being found and employed via the vetting that referrals provide.  Thus, this information asymmetry further amplifies the inequality, and also can further impact overall productivity.  

\paragraph{Policy Implications}
\

Modern platforms that connect employers with workers enable new policies that can counteract the inequities and inefficiencies that we document in the model. In particular, we show that a form of algorithmic fairness on the part of a platform (the algorithm) can improve overall efficiency and equality. Workers who do not have any referral connections not only suffer from missing out on the early hiring option by firms but also end up in a pool together with workers who had referrals and were rejected.  Thus, workers who do not have referrals are also hurt by a lemons effect of being lumped together with lower-quality, previously rejected workers.  If a platform or application process can identify which workers from the secondary pool have already had opportunities via referrals and which have not had any previous chances, then that eliminates the lemons effect and improves the attractiveness of hiring previously unreferred workers from the pool.  Overall, such a policy is very simple and has a number of positive effects: it reduces uncertainty and inequality and improves match quality and productivity.  This policy also reduces immobility since more of the disadvantaged group are hired, which then has a positive multiplier across generations.

Next, we consider more standard affirmative action policies from a new angle.  Such policies boil down to influencing the relative fractions of workers hired from groups with low historic employment (the ``greens'' in our model, as a stand-in for such groups, which could be defined by gender, caste, ethnicity, etc.) and groups with high historic employment (the ``blues'').  
Importantly, our results on productivity imply that one should care about imbalances in access to referrals not just because of inequality and immobility (i.e., fairness), but also because imbalances reduce efficiency.  By bringing the balance of employment rates across groups into better alignment, referrals are more widely distributed across the population, giving more people chances to be hired, lowering overall search frictions, and improving efficiency. 
Also, affirmative action is usually analyzed in terms of short-run implications.  Instead, we show that it has long-lasting implications for inequality, immobility, and productivity.  The basic intuition is that it helps equalize the distribution of referrals across groups, and that this network effect then persists across generations.  This offers an explanation as to why data show that affirmative action has a lasting effect, even after it is removed (e.g., see \cite{miller2017}, and for more background on affirmative action, see \cite{holzern2000,fang2011}).  This provides a very different view on affirmative action than one from focusing on the perspectives of discrimination or of asymmetries in opportunities for education (e.g., see \cite{holzern2000,fang2011}).

In the appendix, we derive further implications of the model, including:  (i) how different approaches to implementing affirmative action have different costs and benefits, (ii) how laws regarding how easily and quickly firms can fire workers change the search frictions and affect all three of inequality, immobility, and inefficiency,  and (iii) how macroeconomic conditions on the number of firms hiring affect inequality, wages, and efficiency of the equilibrium.

\paragraph{Challenges in the Analysis}
\

Our analysis is more complicated than initial intuitions may suggest.  The main challenges are twofold.
First, there are two different ways in which referrals can be concentrated among part of a population, and each has different effects. To understand this distinction it is important to note that more than one current worker might refer the same worker from the next generation. For instance, if blues tend to refer other blues, then as the current employment becomes more tilted towards blues, more blues will get referrals and fewer greens will---with more blues getting multiple referrals. This suggests two forms of referral concentration: the fraction of a population receiving any referral at all may be higher in one group than the other, and the fraction of a population receiving more than one referral may be higher in one group than the other. Given that referrals provide information, the first type of referral concentration leads firms to have information about fewer workers, which hurts overall productivity.\footnote{In terms of the distribution of productivity, it helps the productivity of employed blue workers since a greater fraction of them are vetted via referrals, and hurts that of greens since fewer of them are vetted.} Thus, it is this sort of concentration that is the key to understanding the impact of changing the distribution of referrals on productivity. In contrast, to understand inequality in wages, one needs to understand the effects of multiple referrals. Having more than one referral gives a worker more potential offers and, hence, more bargaining power. Thus, determining the fraction of the population that gets higher wages as referrals are concentrated depends on the fraction of the population that gets multiple referrals. Note that increasing one type of concentration is neither necessary nor sufficient for increasing the other.\footnote{For instance, suppose that the population of applicants consists of people who have zero, one, two, or three referrals.  By redistributing the extra referrals from those who have three to the rest of the population, it is possible to both increase the number of people who have multiple referrals at the same time as the number of people who get any referrals. Thus, it is possible to have a larger fraction of the population getting multiple referrals, while still having more of the population get referrals, and keeping the overall total constant.  We provide conditions under which both concentrations move together.} Since how many workers get any referral governs productivity, and how many workers get multiple referrals drives wage inequality, there is a close but imperfect relationship between how productivity and inequality are influenced by referrals.

The second challenge in our analysis is that beyond the direct effect that referrals have on who gets employment offers, they also have an indirect effect---a lemons effect---that makes comparative statics regarding the impacts of referrals subtle. The lemons effect refers to the phenomenon that some workers in the pool had received referrals but were not hired via those referrals---so they were already rejected by at least one firm. The fact that the pool includes previously rejected workers lowers the expected productivity of workers hired from it.   The main complication in our analysis is that the lemons effect decreases as one concentrates referrals, since there are then fewer workers who are screened by hiring firms and rejected. This effect attenuates the relationship between productivity and the concentration of referrals. It also complicates the relationship between inequality and the concentration of referrals and causes comparative statics on inequality to differ fundamentally from the comparative statics on productivity.  All of our results, and their proofs, deal with this issue.  

We remark that our main results still hold if firms can perfectly observe the type of the worker that they are facing in the pool.  In fact, even the lemons effect remains with perfect observability of a match from the pool.  Thus, it is not an information asymmetry that drives the results. The key is that there is an advantage to workers who are referred, as they have extra chances of being hired relative to those in the pool. The main force in the model is a search friction captured through the limited number of draws that a firm gets to hire from.  As long as there is some friction in search, workers who can be seen by firms more times have an advantage. This brings implications of search frictions (e.g., see \cite{RogersonShimerWrigt2005} for a literature review) into a network setting with homophily, showing that they have strong implications for persistent inequality.     

\paragraph{Relationship to the Literature}

\

There is consistent evidence that referred workers tend to be more productive per unit of time, more inventive, and stay longer in their positions than non-referred workers.  For a variety of evidence regarding these facts---as well as the fact that referrals provide information about the potential productivity of the worker and how ties between workers (and managers) affect productivity---see  \cite{fernandezcm2000,bandiera2008social,bandiera2009social,bandiera2010social,brownst2016,fernandez2014causal,burks2015,dustmann2016,pallaiss2016,heath2018,jackson2019,bondf2019,benson2019discrimination,dhillon2021employee,friebelhhz2023}. This is a result that falls out of our model, given that referred workers are only hired when they are of higher expected value than the available pool.  

Previous papers have shown that the reliance on referrals can result in inequality in employment and wages across groups, both empirically and in theory.  For some of the extensive empirical evidence on this and further references, see \cite{munshi2003,arrowb2004,calvoj2004,calvoj2007,beaman2012,patacchiniz2012,laschever2013,beamankm2016,lalannes2016,jackson2019,arbex2019network,zeltzer2020,okafor2020social,jackson2021}.    The relationship between referrals and inequality also falls directly out of our model, since referred workers have extra chances of being hired, and higher average wages conditional upon being hired, given that their value is known and they are being selected for that higher value. In addition, the fact that search processes can lead to lemons effects has also been studied before \citep{gibbonsk1991,farber1999,kruegercc2014,bondf2019}. 

The most relevant theoretical work to ours is \cite{montgomery1991}.  \cite{montgomery1991} examines a model in which workers' productivities correlate with their connected workers' productivities.\footnote{See \cite{galenianos2021referral} for a justification of such a productivity homophily.}  This gives firms a reason to hire via referrals, leading to higher wages for referred workers and a lemons effect among unreferred workers.\footnote{Other research in which informational asymmetries lead to better outcomes for some workers include \cite{waldman1984,milgromo1987,conde-ruizgp2021}.} 
Our model also results in these two effects but for somewhat different reasons.  Notably, in \cite{montgomery1991}, the lemons effect is exogenous: unreferred workers have lower productivity, while in ours, the lemons effect is endogenous because of firm decisions and occurs even though unreferred workers have the same ex-ante average productivity as referred workers. 
Furthermore, in \cite{montgomery1991}, the lemons effect and wage inequality get worse with increased correlation between referrer and referred values, a worker-quality type of homophily, whereas in our model wage inequality gets worse with increased group homophily.  Most importantly, in \cite{montgomery1991}, there is always full employment, and so there are no questions about productivity, unemployment inequality, mobility issues, nor any of the policy questions that we tackle. These are the main focus of our analysis.

The main other model that examines homophily and referrals is \cite{buhai2023social}.\footnote{See \cite{okafor2020social} for a variation of the \cite{montgomery1991} model with homophily that is used to calibrate wage differences as from \cite{arrowb2004}.} They examine endogenous education and occupation choice and segregation and implications for inequality, which we do not, while we examine referrals as a complement to an open application market and resulting inequality and its dynamics within a single occupation. 

There are other models of job networks that explore other questions tangentially related to the questions we examine here.  For instance, \cite{calvoj2004,calvoj2007} show how referrals correlate employment between connected workers and examine dynamic incentives to invest and a resulting poverty trap.  Thus, immobility in their model comes from an investment decision rather than the dynamics of the referral distribution.  Also, they do not model firm behavior and thus do not explore productivity or wage distributions.  Nonetheless, we discuss how such a poverty trap can be important in ensuring that the market does not eventually converge to full equality. \cite{calvo2004,galeottim2014,galenianos2021referral} illuminate the endogeneity of job networks, which we come back to in the concluding remarks. 

\cite{dustmann2016} show that referrals lead to differential employment length and wage progression of minority and majority workers at a given firm. \cite{gibbonsk1991} explore the impact of layoffs on the lemons effect in the unemployed pool, demonstrating both theoretically and empirically that laid-off workers have longer periods of unemployment and lower post-displacement wages.
Although these and other models point out how referrals can generate inequality,  as well as correlations between workers' outcomes, {\sl none} of our main questions---jointly examining referrals' impact on inequality, immobility, and inefficiency, as well as studying affirmative action and other policies---are addressed by these models. 

Finally, we relate to a literature studying the use of algorithms to assist human decision-making, particularly in hiring---e.g., to screen open applications for suitable candidates  \citep{lirb2020hiring_as_exploration}. 
Since humans can make biased hiring decisions \citep{quillian2017meta}, the hope is that algorithms may not only help sort through massive amounts of data but also provide relatively less biased assessments \citep{houser2019ai, kleinberg2017human_machine}, although the use of algorithm also requires careful supervision and evaluation \citep{raghavanbkl2020evaluating_algorithms}. 
In our setting, even algorithms screening for suitable candidates from open applications may be biased \emph{if they do not take a candidate's network into account}. 
Our model shows how an algorithm (or a job board) screening and recommending candidates from open applications improves efficiency and reduces inequality by conditioning on whether an applicant had previous opportunities to be hired. 

\section{A Model and Preliminaries}\label{sec:model}

We consider a labor market with a unit mass of risk-neutral firms, each having one current but retiring worker. Each firm wishes to hire at most one next-generation worker to replace its retiring worker.%
\footnote{Although we model agents as just working for one period, this is for convenience, and one can think of a period as a hiring season.} There is a mass $n\geq 1$ of risk-neutral workers seeking work at these firms; so there can be unemployment. We refer to these agents as ``workers'' regardless of their employment status.

A generic worker $i$ provides a productive value $v_i$ to any firm that employs that worker. This value includes the worker's skill and talent, and whatever else makes the worker productive, and is the maximum amount that the firm would be willing to pay to hire the worker if the firm had no other possibilities of filling the vacancy. The productive value $v_i$ is distributed according to some distribution $F$ and independent across workers.\footnote{See \cite{currarinijp2009} for background references and a discussion about the ways in which one can justify the independence here and the random matching with a continuum of people.} $F$ has a finite mean, and we consider the nondegenerate case in which $F$ has weight on more than one value, and so $v_i$ has nonzero variance.

Firms hire workers either through referrals or via a pool of open applications---via ``referrals'' or  ``the pool.'' Referrals are generated through a network in which each firm's current (retiring) worker refers a next-generation worker. For simplicity, we assume each (employed) worker refers exactly one worker. 

Nonetheless, some next-generation workers may get multiple referrals as several current workers may each refer the same next-generation worker to their respective firms. We denote the distribution of referrals a worker gets by $P$, with $P(k)$ being the probability that a generic worker gets exactly $k$ referrals.\footnote{This allows for a variety of different referral processes, each captured via a different degree distribution. For example, if referrals are made uniformly at random across all workers, then $P$ is a Poisson distribution; i.e., $P(k) = \frac{n^{-k} e^{-1/n}}{k!}$, with $1/n$ being the mean number of referrals.}

We assume throughout that $1>P(0)>0$ where the first inequality rules out the trivial case in which no referral market exists, and the latter has to be the case if, e.g., $n>1$.

Workers have a minimum wage that they must be paid, $\ubar{w}$, that is presumed to be equal to their outside-option value. Our results do not change substantially if $\ubar{w}$ is greater than workers' outside-option value, and we discuss differences as they arise. We presume throughout that $\ubar{w}$ lies below the max of the support of $F$ so that there is a positive mass of workers that firms find strictly worth hiring.

Each firm observes the value of its referred worker and then chooses what, if any, wage offer to make. When multiple firms are competing for the same referred worker, they begin bidding for the worker until no firm wants to increase its bid.\footnote{We thus implicitly assume that firms know whether a worker has another referral.  Note that it is in the worker's interest to let the firms know that they are competing, and would result in exactly the wage we use if the worker can go back and forth between firms showing them each other's offers. We do not require this assumption; any process in which the wage of workers is increasing in the number of referrals suffices.

We also note that other wage-determination processes can be incorporated into the model. What is key for our results is not the particular wage structure but instead that more firms competing for the same worker leads to a higher wage.} If a firm chooses not to hire its referred worker, it can go to an anonymous pool, consisting of all workers who either have no referrals or are not hired via any of their referrals, and may hire one worker picked at random from that pool. The fact that a firm just gets at most one draw from the pool is a simple way of modeling a search friction capturing that it is time-consuming to learn the workers' values, or because workers are difficult to even find.\footnote{One could instead use a Diamond-Mortensen-Pissarides-style model for the dynamic frictions of finding a match in the pool.  What is essential is that referred workers have some extra opportunity to be matched and that there is some information in hiring.  If there was no information about workers through referrals nor through the open application, then it would not matter at all how many draws firms had at hiring a worker and all hires would have equal expected values.}

We note that although we assume that workers always refer a next-generation worker to their employer, regardless of that worker's value, our analyses of the base model (Sections~\ref{sec:model} and~\ref{sec:groups}) follow with no changes if current workers only refer friends whom they know will be hired in equilibrium. 

Firms get a payoff equal to the expected value of the worker they hire (if they hire one) minus the wage.  A worker's payoff is the wage if they are hired and their outside option otherwise.

In summary, the timing of the game is as follows.
\begin{enumerate}
\item Hiring from Referrals:

\begin{enumerate}
\item Each firm has one worker referred to it, observes the value of this referred worker, and chooses whether to make an offer to its referred worker and, if so, at what wage.
\item Referred workers who receive at least one wage offer choose to accept one of them or to reject all of them.
Any accepted offer is consummated, and the firm and worker are matched. 
\end{enumerate}
\item Hiring from the Pool:
\begin{enumerate}
\item Workers and firms who are unmatched after the referral stage go to the pool.  Each such firm gets matched with one worker chosen uniformly at random from the pool (without replacement, so that no two firms are matched to the same worker from the pool).
Firms choose whether to make an offer to its worker from the pool based on the expected value of the worker and, if so, at what wage.
\item Workers from the pool who receive a wage offer accept or reject it.
Any accepted offer is consummated, and the firm and worker are matched.  
The remaining firms and workers are unmatched.
\end{enumerate}
\end{enumerate}

\subsection{Equilibrium Characterization}\label{sec:one-period}

We examine (weak) perfect Bayesian equilibria of the game.

The basic equilibrium structure is easy to discern and can be seen from backward induction. A worker accepts any wage that is at least the minimum wage, $\ubar{w}$, since that is also their outside option. The only possibility to have workers mixing at $\ubar{w}$ is if both firms and workers are indifferent.\footnote{Otherwise, there cannot exist an equilibrium in which workers mix since if workers mixed and firms were not indifferent, then firms would deviate and offer a higher wage.} In that (non-generic) case, the mixing becomes irrelevant since there is 0 net expected value in the relationship to either side.\footnote{Mixing can matter in the relative employment rates of different groups of workers, and so we track mixing when we get to that analysis.} Thus, equilibrium behavior in the pool stage is such that firms hire workers from the pool at a wage of $\ubar{w}$ if the expected value of workers in the pool exceeds $\ubar{w}$, do not hire from the pool if the expected value is below $\ubar{w}$, and both sides can mix arbitrarily if the expected value of workers in the pool is exactly $\ubar{w}$ (which is also the workers' outside option value).

Then, given the continua of firms and workers, taking the strategies of others as given, firms have a well-defined value from waiting and hiring (or not) from the pool---either something positive or 0. Thus, in the first period, a firm prefers to hire a referred worker if and only if that worker's value minus the wage exceeds the expected value from the second-period potential pool hiring. If there is no competition for the referred worker, the worker's wage is $\ubar{w}$; otherwise, the competing firms bid until they are indifferent between hiring and not hiring the worker.

The key to characterizing the equilibrium is thus a threshold such that firms attempt to hire a referred worker who has a value above that level and do not hire workers below that level. That threshold corresponds to the value of waiting and possibly hiring from the pool. To characterize the threshold, the relevant function is the expected value of workers in the pool conditional on firms hiring referred workers who have values strictly above ${\tilde{v}}$, not hiring workers with values strictly below ${\tilde{v}}$, and hiring workers with values exactly equal to ${\tilde{v}}$ with probability $r$:\footnote{Neither firms nor workers have to follow similar strategies at the threshold, $r$ simply represents the overall probability that a worker with exactly the threshold value is hired overall and includes the mixing of both firms and workers. For simplicity, we refer to $r$ as the mixing parameter.} 
\begin{equation}
\label{tau}
\E_{\tilde{v},r} [v_i|i\in \pool]  \coloneqq   \frac{ P(0) \E[ v_i  ]  +  (1-P(0)) (\Pr(v_i<{\tilde{v}})\E [v_i  |  v_i< {\tilde{v}} ]+\Pr(v_i=\tilde{v} )(1-r)\tilde{v}) }{P(0) +(1-P(0)) (\Pr(v_i<\tilde{v})+\Pr(v_i=\tilde{v})(1-r))}.
\end{equation}
So, all equilibria are equivalent to using a threshold $\tilde{v}$ (and a mixing parameter $r$) that is based on a fixed point of \eqref{tau}; i.e., a threshold giving rise to an expected value in the pool making the threshold optimal for firms.\footnote{For discrete distributions $F$, there could be a multiplicity of thresholds that all correspond to the same decisions.} Thus, we use the term ``equilibrium threshold'' to refer to a fixed point of \eqref{tau}.

\begin{lemma}
\label{eqw}
There is a unique threshold solving
\begin{equation}\label{eqw_con}
\tilde{v} = \max \{ \ubar{w},  \E_{\tilde{v},r} [ v_i  |  i \in \pool] \},
\end{equation}
and that threshold value characterizes the following equilibrium behavior:
\begin{enumerate}
	\item a referred worker $i$ is hired if the worker's value $v_i>{\tilde{v}}$, not hired if $v_i<{\tilde{v}}$, and hired with any arbitrary probability if $v_i={\tilde{v}}$;\footnote{Note that if $(\tilde{v},r)$ satisfies \eqref{eqw_con}, then so does $(\tilde{v},r')$ for any $r' \in [0,1]$}
	\item firms that are unsuccessful in hiring a referred worker hire from the pool if $\E_{\tilde{v},r} [ v_i  |  i \in \pool]> \ubar{w}$, do not hire if $\E_{\tilde{v},r} [ v_i  |  i \in \pool] < \ubar{w}$, and hire from the pool with any arbitrary probability if $\E_{\tilde{v},r} [ v_i  |  i \in \pool] = \ubar{w}$.
\end{enumerate}
The equilibrium wage of a hired worker $i$ is $v_i-{\tilde{v}}+\ubar{w}$ if $i$ has more than one referral, and $\ubar{w}$ otherwise.
\end{lemma}

Although we account for all cases in what follows, the reader may find it easier to concentrate on situations in which firms find it worthwhile to hire workers from the pool, since then $\tilde{v}= \E_{\tilde{v},r}[v_i|i \in \pool] $ and the threshold is simply the expected productivity in the pool.

The remaining details behind the proof of Lemma~\ref{eqw} (beyond the discussion that precedes the lemma), including why the threshold is unique, and all subsequent proofs of lemmas and propositions appear in Appendix~\ref{appendix:proofs}.

Given that hiring on the referral market follows a simple threshold rule, that is, workers with referrals are hired if their productivity is above some threshold, it follows that rejected workers go to the pool, which lowers the average value of workers in the pool. The impact of this selection on the pool productivity is a sort of lemons effect common to search markets.

\begin{lemma}
\label{lemma:LE}
In equilibrium there is a (strict) lemons effect:
$\E_{\tilde{v},r} [ v_i  |  i \in \pool]  <   \E [ v_i ]$.
\end{lemma}

We track the lemons effect in our analysis, as it can either amplify or mitigate the various effects that we study. For instance, it provides additional incentives for firms to hire referred workers. This feedback means that the hiring threshold for referrals, ${\tilde{v}}$, unless it is determined by the minimum wage, is less than the unconditional expected productivity value in the population,  giving a further advantage to referred workers. Thus, workers who have referrals not only have an additional chance to be hired compared to those who are only in the pool but also benefit from the lemons effect, which makes firms even more willing to hire via referrals.

\section{Worker Groups and the Impact of Referrals on Outcomes}\label{sec:groups}

With the basics of the model in place, in this section, we introduce different groups of workers; e.g., by tracking ethnicity, gender, age, education, geography, etc. Then, homophily (tendencies to refer own group) or some other asymmetry (people being relatively biased towards referring some particular group)  lead to differences in referrals across groups. Those differences have direct inequality and inefficiency implications; but also affect the dynamics.  Tracking groups across time enables us to see how low current employment among a group translates into low employment and low wages in that group's next generation---immobility.

For simplicity, we consider two groups, but the results extend easily to more. We refer to one group as {\sl blues} and the other as {\sl greens}, with respective masses $n_b>0$ and $n_g>0$ of workers per generation, such that $n_b+n_g=n$. Let $e_b$ and $e_g$ be the masses of employed blue and green workers at the beginning of the period, respectively; i.e., the masses of the retiring workers or the ``current employment'' for short. Without loss of generality, we maintain that $\frac{e_b}{e_g}\geq \frac{n_b}{n_g}$ and say that there is a current employment bias towards blues if the inequality is strict.

\subsection{Homophily and Referral Distributions}\label{sec:homophily}

To model homophily in referrals, we track group-dependent referral-bias parameters $h_b\in [0,1]$ and $h_g\in [0,1]$. A fraction $h_b$ of employed blue workers of the current generation refer blue workers from the next generation, and the remaining $(1-h_b)$ fraction of employed blue workers refer green workers, with $h_g$ defined analogously. Levels of $h_b>\frac{n_b}{n}$ and $h_g>\frac{n_g}{n}$ indicate homophily; i.e., a bias towards referring one's own group. (An alternative form of homophily is on values, which is of less direct interest here, but we explore in Section~\ref{corrvalues} of the appendix. Allowing for such homophily does not alter our qualitative results.) Employment levels $e_b,e_g$ and homophily levels $h_b,h_g$ determine the average number of referrals that are made to blues and greens.  For instance, the average number of referrals blue workers get is 
$$
m_b=\frac{h_be_b+(1-h_g)e_g}{n_b},
$$
and similarly for greens (swapping labels). 

We assume that $h_b \geq 1- h_g$.  This is implied by homophily but is, in fact, a weaker condition. It admits, for instance, a situation in which both greens and blues bias referrals towards, say, blues as happens in some cases with gender bias. Instead, it rules out extreme cases with reversals: greens referring blues more often than blues referring blues, and vice versa.  This assumption implies the natural case:  if a group's employment goes up, then it gets more referrals.\footnote{An absence of this assumption can lead to peculiar dynamics: start with high employment among blues who mostly refer greens, which then leads to high green employment who then mostly refer blues, continuing in a cycle.  As we show, even ruling this case out, it is still possible to cycle, but those cycles are more plausible.} Indeed, our results (suitably adjusted) extend under the weaker assumption that $m_b$ and $m_g$ are increasing in $e_b$ and $e_g$, respectively, but not necessarily linearly.

Recall that the overall distribution of referrals is denoted by $P$; here, we decompose this distribution and suppose that referrals for a group happen according to some distribution conditioned on $m$, the average number of referrals for that group: $\widehat{P}(\cdot|m)$. Thus, the overall distribution of referrals is given by $P(\cdot|m_b,m_g)=\frac{n_b}{n}\hat{P}(\cdot|m_b)+\frac{n_g}{n}\hat{P}(\cdot|m_g)$.  

We intentionally consider the case in which the distribution of referrals is the same for both groups---mean adjusted---which thus isolates the effects of homophily. Let $\hat{P}(2+|m)\coloneqq \sum_{k\geq 2} \hat{P}(k|m)$; we also focus on the case such that 
\begin{enumerate}
    \item \label{a1} for $m'>m$, 
$$
    \hat{P}(0|m')<\hat{P}(0|m) \quad \text{and}\quad
    \hat{P}(2+|m')> \hat{P}(2+|m)\text{; and}
$$
    \item \label{a2} $\widehat{P}(0|\cdot)$ is strictly convex.
\end{enumerate}
Condition~\ref{a1} is natural and states that more referrals for a group decreases the fraction of workers in that group with no referrals and increases the fraction of workers with multiple referrals. Condition \ref{a2} states that more referrals in a group decrease the fraction of workers with no referrals at a decreasing rate. All these conditions are satisfied if each worker makes their referral according to the same stateless process, e.g., uniform at random, in which case $\widehat{P}(\cdot|m)$ is a Poisson distribution.

Note that some initial employment across groups in period $t$ determines the referral distribution for workers in period $t+1$, which in turn determines the employment across groups in period $t+1$, etc. 

To keep the analysis uncluttered, firms maximize their profits in each time period separately. The equilibrium (in each period) can then be analyzed as discussed in Section~\ref{sec:model}, since whether a worker is green or blue does not impact a firm's immediate payoffs. We comment in Section~\ref{sec:forward} on how the results adjust if firms account for the value of a worker's future referrals, which differ by group.

\subsection{Inequality, Immobility and Inefficiency}\label{sec:ineq_and_immobility}

In this section, we show that homophily, together with inequality in current employment, leads to inequality (in both wages and employment) and inefficiency (in terms of total productivity) in the current and future periods.  Here, we examine inequality between groups, and in Section~\ref{sec:Pchanges} of the appendix, we discuss overall inequality, immobility, and productivity across all members of the society as a function of the overall referral distribution, of which this is a special case.

To understand these impacts, it is useful to consider how many referrals workers would get if employment numbers were equal to the relative population proportions. If there was no current employment bias, then the blue and green workers looking for jobs would get a total of
\begin{equation}\label{rbrg}
R_b=h_b n_b+(1-h_g)n_g \quad \text{and} \quad R_g=h_g n_g+(1-h_b)n_b
\end{equation} 
referrals respectively. We say referrals are {\sl balanced} if each population receives a number of referrals proportional to their representation ($\frac{R_b}{R_g}= \frac{n_b}{n_g}$); e.g., when referrals are purely homophilous ($h_b=h_g=1$) but also if referrals are independent of groups ($h_b=\frac{n_b}{n}$,$h_g=\frac{n_g}n$). 

If current employment rates are balanced ($\frac{e_b}{e_g}= \frac{n_b}{n_g}$) and referrals are balanced ($\frac{R_b}{R_g}= \frac{n_b}{n_g}$), then outcomes are equal for the two groups, regardless of the degree of homophily. On the other hand, should either of these conditions strictly favor one group, the blues ($\frac{e_b}{e_g}\geq \frac{n_b}{n_g}$ and $\frac{R_b}{R_g}\geq \frac{n_b}{n_g}$ with at least one strict inequality), then the referral distribution of blue workers dominates that of green workers in that larger fractions of blues get at least one referral and multiple referrals, resulting in higher employment rates and better wages, and this effect persists over time resulting in immobility (although the inequality can subside over time as discussed in Section~\ref{sec:dyn}).

This stochastic dominance of the blue worker referral distribution also drives the inefficiency of equilibria.  We define the efficiency or productivity of the equilibrium to be the total value of employed workers, plus the outside option $\ubar{w}$ for all the unemployed workers.  Any equilibrium is clearly inefficient as firms hire referred workers with productivity below other workers in the pool (due to the search frictions when hiring from the pool).  Nonetheless, one can show that an equilibrium is still {\em constrained efficient}---the threshold for hiring refereed workers that maximizes total production is the unique equilibrium threshold (see Section~\ref{sec:constrained_efficiency} of the appendix).  However, although each equilibrium uses the best threshold it can, the equilibrium efficiency degrades as referrals become more imbalanced.  To understand why, it is useful to note that a single referral is productivity enhancing: if a referred worker has high productivity (above the equilibrium threshold), then that worker is hired.  If the worker has low productivity, then the firm has a chance of hiring a high-productivity worker from the pool.  When some worker gets two or more referrals instead of one, then that does not improve the matching of that worker beyond the first referral--if they are low value, then the referrals are all wasted, and if they are high value, then any referral past the first one is wasted.  As referrals become more balanced, they are spread more evenly across the population, reducing the chance of such collisions while also improving the mass of workers with at least one referral.  

This discussion is summarized in the following proposition.

\begin{proposition}
\label{proposition:group_outcomes}
If there is a (weak) employment bias ($\frac{e_b}{e_g}\geq  \frac{n_b}{n_g}$) and a (weak) referral imbalance ($\frac{R_b}{R_g}\geq \frac{n_b}{n_g}$) in favor of blues, then
\begin{itemize}
\setlength\itemsep{-0.5em}    
	\item the wage distribution of blue workers first-order stochastically dominates the wage distribution of green workers (wage inequality),
	\item the employment rate of blue workers is (weakly) higher than the employment rate of green workers (employment inequality),
	\item and this is true of all future periods (immobility).
\end{itemize}
Furthermore, if referrals are purely homophilous---i.e., $h_b=h_g=1$---then
\begin{itemize}
    \item employment ratios that are closer to being population-balanced strictly increase productivity\footnote{That is, if $\frac{e_b'}{e_g'} > \frac{e_b}{e_g} \geq \frac{n_b}{n_g}$, then the next-period productivity in the equilibrium with current employment $(e_b',e_g')$ is strictly less than the next-period equilibrium productivity with current employment $(e_b,e_g)$.} (efficiency).
\end{itemize}
\end{proposition}

We note that the average productivity of employed blue workers is at least as high as that of employed green workers, since more of them are selected via referrals. Correspondingly, the average productivity of unemployed green workers is at least as high as that of blues. All comparisons above are strict if at least one of the bias and imbalance conditions holds with strict inequality.

There are some subtleties in the proof, particularly of the efficiency result: unequal referrals across groups, and hence a higher aggregate probability of not getting a referral, implies that there are fewer lemons in the pool, making the pool more productive and raising the equilibrium threshold. The proof shows that this countervailing force is always overcome. The basic intuition is that having a higher probability of not getting a referral means that fewer workers get vetted overall, and more ultimately end up being hired without any vetting (lemons or otherwise). Ultimately, jobs that are not filled by a high-value worker end up being filled by someone of lower value or not filled at all (depending on the equilibrium, presuming that we are not changing from one type of equilibrium which hires from the pool to the other which does not), and replacing those with high value is good. The full proof takes care of all the possible cases, including those comparisons where workers are hired from the pool for one but not the other employment ratios.

One interesting implication of the conditions in Proposition~\ref{proposition:group_outcomes} is that blues suffer a worse lemons effect than greens since relatively more of them are screened and rejected. Hence, the average value of blues in the pool is worse than that of greens. As a higher fraction of employed blue workers is hired through the referral market compared to greens, employed blue workers have higher productivity than employed green workers.\footnote{This is consistent with the observation that referral-hired workers tend to outperform other workers; and employed blues are more often referrals.} While beyond the scope of the model, the difference in average productivity of employed workers across types could perpetuate a biased perception of their respective abilities if observers (human or algorithmic) do not understand the selection process and estimate the population mean directly using the sample mean (i.e., observers assume that ``what you see is all there is,'' a term coined by Daniel Kahneman in, e.g., \cite{Kahneman2011}). In particular, estimating productivity by looking at employed populations (the way that productivity is usually measured) systematically overestimates blues' productivity and underestimates greens' productivity and may lead to further (inaccurate) statistical discrimination.\footnote{\cite{bohrenhip2019} highlights that such inaccurate statistical discrimination may be indistinguishable from taste-based discrimination.}

\subsection{Convergence to Steady State}\label{sec:dyn}

Proposition~\ref{proposition:group_outcomes} shows that inequality persists over time, but what exactly is the trajectory?  In this section, we show that, so long as the minimum wage is low enough so that firms find it worthwhile to hire from the pool in steady state, the dynamics will converge to a unique steady state.\footnote{Without this condition, it is possible to get cycles (see Example~\ref{ex:cycle_pool} in Section~\ref{online:examples} of the appendix for an example).} This helps us analyze policy interventions (see Section~\ref{sec:aa}). If one group is more self-biased in giving referrals than the other group, then that group gets relatively more referrals, and they will dominate employment in the long run.  Convergence to balanced employment thus requires a balance in referral rates. 

\begin{lemma}
\label{lem:longrun}
There exists a unique steady-state employment rate for each group. The steady-state employment rates are balanced ($\frac{e_b}{e_g}= \frac{n_b}{n_g}$) if and only if there is referral balance ($\frac{R_b}{R_g}= \frac{n_b}{n_g}$). If there is referral balance, then convergence to the steady state occurs from any starting employment levels.\footnote{If referral balance fails, then one can construct examples with non-convergence; see Example~\ref{ex:cycle_bias} in Section~\ref{online:examples} of the appendix.
\label{footex2}}
\end{lemma}

Existence of a steady state follows from a standard fixed-point argument.  Uniqueness is more subtle and depends on bounding the slope of how next-period employment of a group can grow as a function of increasing current employment of that group.  If that ``slope'' is everywhere less than one, then there can be only one fixed point. The idea is as follows.  Adding one extra green today leads to at most one more green referral.  So, the direct effect is at most one.  When greens receive relatively fewer referrals, raising green employment makes the distribution of referrals more spread out which worsens the lemons effect. This lowers the threshold and leads to relatively more blue hires on the referral market, working against green hires.  When greens receive relatively more referrals, then increasing green employment now decreases the lemons effect and raises the threshold---now disadvantaging the greens again since now they are getting more referrals. Thus, the ``slope'' of next-period green employment as a function of today's employment is at most one establishing the uniqueness. Balanced employment rates are a fixed point if referrals are balanced and not otherwise, and so uniqueness implies that this is the steady state if (and only if) balance holds.

Convergence is again more subtle.  Balanced referrals can be shown to imply a monotonicity, so that a group that is underemployed gains employment but never more than to a balanced level. However, without referral balance, the indirect effect can dominate and lead to a situation where the change in employment overshoots the steady state.  

Lemma~\ref{lem:longrun} shows that if referrals are balanced, then initial employment rates become irrelevant in the very long run.  Of course, that is over generations, and so with high rates of homophily, inequality could persist for many generations.

\subsection{Costly Investment and Immobility}
\label{sec:investment}

Lemma~\ref{lem:longrun} implies that if referrals are balanced, inequality disappears asymptotically. Here, we show that if we add costly investment to the model, then such convergence can fail even with balance, i.e., there is permanent (nontrivial) inequality.  

Our model of referrals can be seen as a signal extraction problem. The values of the advantaged group are more likely to be seen because of the better referral network.\footnote{Even in the absence of costly investment, such environments can lead to poor outcomes for the group with noisier signals \citep{phelps1972,fang2011}.} Individuals whose values are less likely to be seen (so their signals are effectively noisier) have lower incentives to make costly investments \citep{lundbergs1983}. This holds in our model. In addition, our environment is dynamic and shows how less incentives for the disadvantaged group to invest can lead to a less informative process in the future, which perpetuates the inequality, immobility, and inefficiency. This is a form of a poverty trap in which a lack of referrals precludes investments, which leads to low employment and hence referrals.\footnote{See \cite{calvoj2004} for a discussion of a poverty trap in a different job-referral setting.}

In particular, if workers must make a costly investment---e.g., in education---to realize a productive value, then perpetual immobility and inequality can ensue. Specifically, Proposition~\ref{proposition:group_outcomes} implies that the expected wage of blue workers from participating in the labor market is greater than that of green workers. Thus, if there is some cost to educating one's self before knowing whether one will have a referral, then for some investment costs green workers will not invest while blue workers will. We show this formally and further analyze their impact on productivity.

Consider a setting in which workers are either of high or low value, $v_H$ or $v_L$. By default, workers are of low value, but by investing at cost $c>0$, workers have a $\rho<1$ chance of becoming high value. The investment decision occurs before workers know the number of referrals they receive; nevertheless, when making this decision, workers do consider their expected networks of referrals. To isolate the role of investment in sustaining inequality, we focus on the case in which referrals are purely homophilous: $h_b=h_g=1$. We examine subgame perfect equilibria of this game, in which the rest of the game is as before, and workers maximize their expected payoff in the ensuing hiring stage (as described in Section~\ref{sec:one-period}) net of any cost of investment.

Proposition~\ref{proposition:group_outcomes} implies that the strictly-advantaged group of blues ($\frac{e_b}{e_g}>\frac{n_b}{n_g}$) have a higher incentive to invest. Consequently, the equilibria must feature all blues investing if some greens do.  Thus, in non-corner equilibria, blues are endogenously the group with a better distribution of values.   This gives them an advantage on the referral market, which in turn can lead to perpetual inequality (immobility).

Figure~\ref{fig:invest} provides an example.  In this illustration, referrals are made again uniformly at random among a group so that $\hat{P}$ is a Poisson distribution, $n_b=n_g=1$, $v_H=1$ and $v_L=0$, and $\ubar{w}=0$ (so that firms hire from the pool). The probability of an investment being successful is $\rho=0.9$, and the cost of investment $c$ is low ($c=0.01$) in the top row and high ($c=0.075$) in the bottom row. 

Workers invest because it can lead to a high wage (if the investment is successful and they receive multiple referrals) and not to improve employment prospects since the outside option equals the minimum wage. This implies that other workers investing decreases the incentive to invest as it lowers the high wage by increasing the value in the pool.  As a result, a fraction of workers may invest in equilibrium.

Panels~\ref{fig:invest_wholow_cost} and~\ref{fig:invest_whohigh_cost} show the mass of workers, by group, investing for a given difference in current employment rates. 

\begin{figure}[!ht]
\centering
\begin{subfigure}{.33\textwidth}
\includegraphics[]{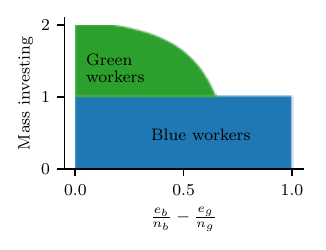}
\caption{}\label{fig:invest_wholow_cost}
\end{subfigure}~
\hfill 
\begin{subfigure}{.33\textwidth}
\includegraphics[]{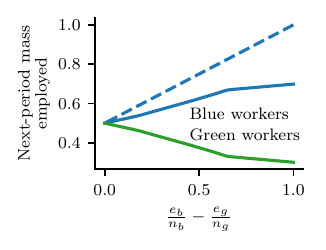}
\caption{}\label{fig:invest_empllow_cost}
\end{subfigure}~
\hfill 
\begin{subfigure}{.33\textwidth}
  \includegraphics[]{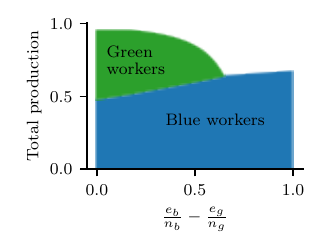}
\caption{}\label{fig:invest_prodlow_cost}
\end{subfigure}
\begin{subfigure}{.33\textwidth}
\includegraphics[]{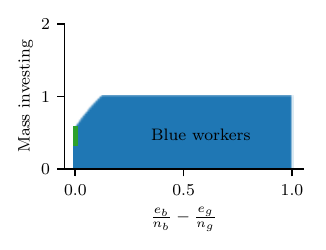}
\caption{}\label{fig:invest_whohigh_cost}
\end{subfigure}~
\hfill 
\begin{subfigure}{.33\textwidth}
\includegraphics[]{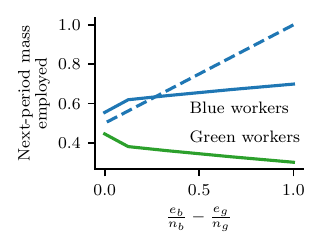}
\caption{}\label{fig:invest_emplhigh_cost}
\end{subfigure}~
\hfill 
\begin{subfigure}{.33\textwidth}
  \includegraphics[]{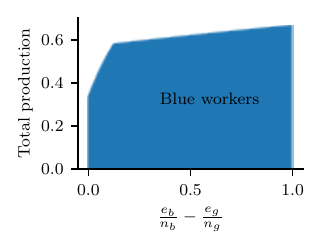}
\caption{}\label{fig:invest_prodhigh_cost}
\end{subfigure}~
\caption{ We plot the investment decision, next-period employment, and total production by group as a function of the initial employment advantage of blues.   The top row has low cost $c=0.01$, and the bottom row has high cost $c=0.075$. Referrals are according to a Poisson distribution: $\hat{P}(k|m) = \frac{m^{k}e^{-m}}{k!}$; and parameters values are: $n_b=n_g=1,v_L=0,v_H=1,\rho=0.9$ and $c\in \{0.01,0.075\}$.}
\label{fig:invest}
\end{figure}

When the cost is low and the employment bias is small, both groups invest. As the blue advantage in current employment grows, green investment drops, and eventually, only the blues invest (Panel~\ref{fig:invest_wholow_cost}).  When the cost is high, then some blues invest when the employment bias is small, and all blues invest when it is large. Panels~\ref{fig:invest_empllow_cost} and~\ref{fig:invest_emplhigh_cost} depict the resulting next-period employment levels. The blue dashed lines consist of points $(\frac{e_b}{n_b}-\frac{e_g}{n_g},e_b)$. Thus, at the intersection of this line and the blue line, the actual resulting next-period mass of employed blues, is a steady state: perpetual inequality.   When the cost is low, there is no employment bias in the unique steady state; however, when the cost is high, the steady state favors blues. Panels~\ref{fig:invest_prodlow_cost} and~\ref{fig:invest_prodhigh_cost} show the total production (i.e., the sum of values produced by employed workers or their outside option if unemployed) as functions of the difference in employment rates. We show the total production by group. When the cost is low, total production is maximized when the employment bias is low, in which case both groups contribute equally. As the employment bias increases, even though all greens still invest, the blues contribute more as they are more likely to be hired due to having more referrals. Green investments decline, and eventually, no green workers invest, and the contribution of greens is zero. When the cost is high, as they never invest, the contribution of greens is zero throughout. Interestingly, in that case, the total production is maximized when the employment bias is strongest since this maximizes the chance of blue workers, who are the ones who invest, to be hired.

We provide a general statement of this result (Proposition~\ref{prop:inv}) in Section~\ref{sec:online_skill} of the appendix.

\section{Market Design and Policies}\label{sec:aa}

In this section, we consider market design aspects and policies that can alleviate the inequality and immobility that arise in referral networks.  Moreover, as we have seen above, doing so not only improves `fairness' but can also increase overall productivity.

We first explore algorithmic fairness, i.e., how a platform (an algorithm) with firms and workers can be designed to alleviate the negative consequences of referrals (Section~\ref{sec:forward}). In particular, firms may be able to see which workers have already had a chance to be hired via referrals and which have not, depending on the design of the platform.  This provides opportunities to un-referred workers and completely eliminates the lemons effect, which leads to short- and long-term improvements in all aspects of the market:  more equality, mobility, and productivity. We also consider a variation whereby firms cannot see whether someone was previously referred but can distinguish greens from blues; i.e., firms can selectively hire disadvantaged workers from the pool. They have incentives to do so since advantaged workers suffer a more severe lemons effect. 

These policies can have a slow impact when vacancies are primarily filled on the referral market or when search costs in the pool are high, or such policies may be impossible to implement for technological or legal reasons. Thus, we also consider a variety of affirmative action policies---which are changes in the given period's employment mix.  Such policies increase minority employment closer to balanced, but also have dynamic implications that we explore. As we show, the impact of such policies is subtle and may be counter-productive in the short term (Section~\ref{sec:ineq_period_to_period}).  However, a well-targeted policy can have positive long-term impacts on both inequality and productivity (Section~\ref{sec:aadynamic}). Furthermore, the potential short-term negative impacts can be mitigated with careful attention as to how the policy is implemented (Appendix~\ref{sec:optAA}). 

Finally, we also explore the impact of firing workers on inequality, as well as the impact of macroeconomic conditions on the functioning of the market (Appendix~\ref{section:firing_workers}).

Throughout this section, we specialize to the situation in which blues and greens are both completely homophilous; i.e., $h_b=h_g=1$. Furthermore, we focus on equilibria in which firms hire when indifferent ($r=1$).\footnote{This assumption is only consequential when firms do not hire from the pool, as then the equilibrium threshold ($\ubar{w}$) does not adjust to small changes in employment. As a result, firms could make more or fewer hires on the referral market. As long as there are no atoms in the distribution of productivities at exactly the min wage, $\Pr(v_i=\ubar{w})=0$, we do not have to worry about this issue.}  These restrictions simplify the exposition but are not necessary for the results that follow.

\subsection{Algorithmic Fairness and Targeted Hiring from the Pool}
\label{sec:forward}

Let us examine how the market changes when firms can target particular workers when hiring from the pool. We consider two cases: firms observe whether workers were previously referred or firms observe workers' groups. Such targeted hiring is becoming increasingly realistic with the advent of large platforms that connect workers and firms. These platforms facilitate and observe referrals and applications and can flag to a firm which applicants have previously been referred; i.e., they can enable targeted hiring from the pool. Such platforms suggest workers to firms, and vice versa, according to algorithms, and such algorithms can be designed to equalize opportunities by prioritizing workers who have not had previous referrals.

Let us first suppose firms are able to target workers in the pool depending on whether they had previous referrals. In this case, un-referred workers are hired before previously rejected workers (the lemons). And when there are more workers without referrals than firms hiring from the pool, then only those workers are hired and the lemons effect is completely eradicated.  This improves productivity, and also reduces the use of referrals as the pool becomes more attractive, which improves equality and mobility.  

Suppose now firms can target workers by their group identity. Then, firms prefer to hire greens from the pool over blues, and they will only hire blues from the pool if no greens remain.  There is still a lemons effect in this case, but it is smaller than when there is no ability to target workers based on identity (e.g., discrimination regulations). 

\begin{proposition}\label{prop:fair}
Suppose there is an employment bias ($\frac{e_b}{e_g}> \frac{n_b}{n_g}$) in favor of blues, and in equilibrium, there is hiring from the pool (in the absence of algorithmic fairness; i.e., without any targeting). Both if firms can observe whether workers in the pool have been previously referred or if firms can prioritize green hiring from the pool, then, relative to the base model:
\begin{itemize}
    \item green employment increases, and
    \item total production increases.
\end{itemize}
Furthermore, the increase in total production is greater when firms can observe whether workers in the pool have been previously referred rather than their group identity.
\end{proposition}

One may expect that targeting based on group identity leads to a greater increase in green employment relative to targeting based on referral status. This is indeed true for any given hiring threshold on the referral market. However, for any fixed threshold, the value of hiring from the pool is greater when hiring is based on referral rather than group status. Thus, the equilibrium hiring threshold is \emph{higher} in the former setting, which relatively benefits green workers, who receive relatively few referrals; see Example~\ref{ex:algfairness} in Section~\ref{online:examples} of the appendix for details.

Proposition~\ref{prop:fair} states that giving firms the information about which workers have been previously rejected not only increases efficiency, but also reduces the employment bias. 

The above proposition deals with the short-run impact of the policies. Comparing the policies dynamically is ambiguous. There are two reasons for this.  One is that group-targeting can lead to a higher green-to-blue employment ratio, which can spread referrals out more evenly in the next period, which then improves the future functioning of the market.  So, although targeting unreferred workers can be more productive in the short run, evening out the referral distribution in the current period can be more effective at improving productivity in the longer run. The second is that there is also an endogenous change in the future lemons composition in the pool which can lead to some non-monotonicities in the dynamics. 

\subsection{Dynamic Complications}\label{sec:ineq_period_to_period}

To see what the difficulties are in identifying long-term effects, let us show how complicated a simple comparative static can be.

Consider an increase in the current employment of green workers, $e_g$, holding everything else constant and without any policies in place. The probability of not getting a referral for green workers decreases while it increases for blues. If we held the equilibrium hiring threshold constant, this would indeed increase the mass of green workers hired via referrals and decrease it for blue workers, and hence increase green employment overall. Let us call this the \textit{direct effect} of increasing current green employment. However, the threshold for hiring referrals is not fixed. As we increase current green employment to make it more balanced, the population-wide probability of not getting a referral decreases. This exacerbates the lemons effect and thus reduces the referral hiring threshold. Thus, workers with referrals are now more likely to be hired. As there are relatively more blue than green workers with referrals, this \textit{indirect effect} disproportionately benefits blue workers. Depending on the setting, this can overturn the direct effect so that increasing current green employment leads to fewer greens being hired in the next period than would have been hired otherwise. Which of the effects dominates depends on the mass of referred workers near the hiring threshold. If this mass is large, then many referred workers are affected by a change in the threshold and the indirect effect can dominate. The exacerbated lemons effect can also lead firms to stop hiring from the pool which also disproportionately hurts the disadvantaged group.

\begin{figure}[!ht]
\centering
\begin{subfigure}{.4\textwidth}
\includegraphics[]{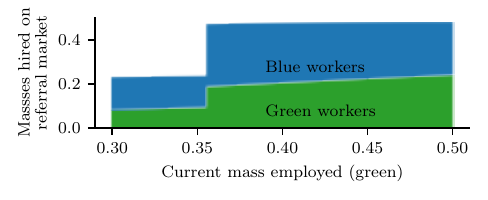}
\caption{Referral market hires}\label{fig:1a}
\end{subfigure}
\hfill 
\begin{subfigure}{.5\textwidth}
  \includegraphics[]{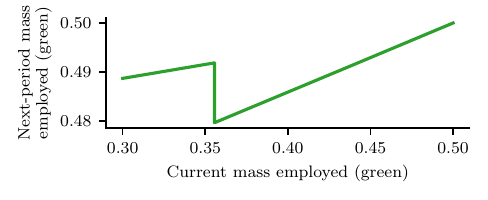}
\caption{Employment of Greens}\label{fig:1b}
\end{subfigure}
\caption{Referrals for a group happen according to a Poisson distribution: $\hat{P}(k|m) = \frac{m^{k}e^{-m}}{k!}$; and parameters values are $n_b=n_g=.7$ so that $n=1.4$; and $ v_i$ takes on three values $0,\frac{1}{3}$, and $1$, with equal probability. We plot the mass of blue and green workers hired on the referral market and the next-period employment mass of green workers as a function of the current employment rate of green workers. }
\label{fig:nonmon}
\end{figure}

Figure \ref{fig:nonmon} shows an example of this phenomenon in which referrals are made uniformly at random among a group so that $\hat{P}$ is a Poisson distribution. Increasing the current mass of employed green workers towards the unbiased employment rate of $1/2$ increases the number of referrals green workers get. The direct effect is then seen via the increase in green workers hired on the referral market. However, additionally the overall probability of not getting a referral decreases, which leads to more lemons in the pool, and so the hiring threshold on the referral market drops. In particular, at a current employment of greens below  $\approx 0.355$, the threshold is such that referred workers with the middle value of $\frac{1}{3}$ are not hired, whereas they are once green employment rises above $\approx 0.355$  (see Panel~\ref{fig:1a}). This change in the equilibrium threshold is the indirect effect: referred workers with the middle productivity value are hired for the higher current employment of green workers; but this disproportionately helps blue workers as they have more referrals. As a result, the next-period employment rate of green workers is lower when their current employment rate is just above $\approx 0.355$ rather than just below $\approx 0.355$ (see Panel~\ref{fig:1b}).

The observed nonmonotonicity may appear due to the discrete nature of our model, but it extends beyond that. There are two sources of discreteness.  One is the discrete distribution. That is not essential to this example---a continuous approximation to this distribution leads to a similar result.  The second is that the market operates in discrete time.  With continuous time, the dynamics depend on the relative slopes of the direct and indirect effects, which still go in opposite directions. Regardless, many labor markets are, in fact, best approximated by a discrete-time model due to their seasonal nature.

\subsection{Affirmative Action's Dynamic Impact}
\label{sec:aadynamic}

The potential complication with affirmative action policies, as we have seen above, is that a change in employment rates can also impact the equilibrium referral hiring threshold, which can result in an indirect effect that can be countervailing.  In spite of that effect, there is much that we can deduce about the dynamic impact of such policies. Indeed, affirmative action not only has immediate effects but can also have longer-term effects since referrals are also moved closer to being balanced. This points out an interesting aspect of affirmative action: It not only increases short-term equality of employment based on characteristics, but it also has long-term network effects that further reduce inequality and improve mobility and future productivity.

First, we show that (generically) a small one-period increase in the employment of the under-employed group increases their employment as well as overall production in all subsequent periods. Thus, affirmative action has long-run implications and even a one-time policy can lead to continued and amplified improvements in both equality and total production.
 
\begin{proposition}\label{aamultiplier}
Let $F$ be any distribution with discrete support. For almost every $e_g$, there exists $\varepsilon>0$ so that an increase in $e_g$ by up to $\varepsilon$ leads to a strict increase in total production as well as green employment in all future periods (relative to what they would have been without the increase).
\end{proposition}

We illustrate these dynamic effects in Figure~\ref{fig:AA}, which shows the impact of a one-time increase in the employment of greens. Referrals are again made uniformly at random among a group so that $\hat{P}$ is a Poisson distribution. We see three long-term effects. First, increasing current employment of greens increases their employment in all periods relative to what it would have been without. Second, the difference in average wages between blues and greens is decreased in all periods. Third, this also improves production in all subsequent periods but has a cost in the current period from replacing referred blues with draws from the pool to get sufficient greens.

\begin{figure}[!ht]
\subfloat[Employment of Greens]{
\label{fig:2a}
\includegraphics[]{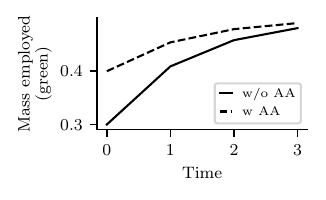}
}\hspace{-1.75em}
\subfloat[Blue Minus Green Wages]{
\label{fig:2b}
\includegraphics[]{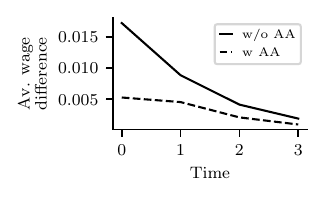}
}\hspace{-1.75em}
\subfloat[Productivity]{
\label{fig:2c}
\includegraphics[]{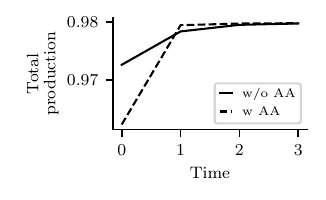}
}
\caption[Caption]{Referrals for a group happen according to a Poisson distribution: $\hat{P}(k|m)=\frac{m^{k}e^{-m}}{k!}$; and parameter values are $n_b=n_g=1$ so that $n=2$; and $v_i$ takes on two values $0$ and $1$, with .95 and .05 probability, respectively. The blue line is without affirmative action and orange is with affirmative action in the first period. The graphs above show the dynamic effects of changing green employment in from .3 to .4 by randomly choosing some firms with high value referred blues and having those firms to go to the pool.\footnotemark{}  
}
\label{fig:AA}
\end{figure}
\footnotetext{As we discuss in Section~\ref{sec:optAA}, having firms forego high-value blue referred workers instead of hiring low-value referred green workers is the more efficient form of affirmative action in this case. Note that we consider a policy that induces enough firms who have blue referrals to go to the pool to effect this change, whether it be via regulation or a payment that compensates the firms forced to go to the pool when having a high-value blue referral.  We remark that blue workers with multiple referrals are less likely than ones with just one referral to be affected by this policy since firms are independently chosen.}

The above example exhibits unambiguously monotone results since it only involves two productivity levels.  With more productivity levels, large changes in employment can lead to threshold reversals, and so large affirmative action policies can involve some non-monotonicites in some future periods. In particular, as shown in Section~\ref{sec:ineq_period_to_period}, an increase in the current employment of a disadvantaged group does not necessarily increase their next-period employment rate compared to what it would have been due to the indirect effect operating through the hiring threshold. As a result, employment rates in some future period can, under some particular situations, be further apart than they would have been without an affirmative action policy (furthermore implying a decrease in total production in that period relative to that without the policy). As shown in Lemma~\ref{lem:longrun}, there would still be eventual convergence, given referral balance, but it could be sped up in some periods and slowed down in others. Thus, care is warranted in designing such policies, given the potential indirect effects due to 
changes in equilibrium lemons effects.  

The long-term impact of affirmative action policies from a one-time intervention is important in light of the issues raised by \cite{coate1993will}. They show that a long-term affirmative action policy can make a group dependent upon the policy and reduce their skill acquisition. Instead, in our model, skill acquisition and affirmative action policies are complements: Increasing green employment through affirmative action policies creates more hiring opportunities for future green workers and incentivizes investment in human capital (see Section~\ref{sec:investment}).

We note that in the context of costly investment, as in Section~\ref{sec:investment}, affirmative action might have to be of a sufficient size to have an impact.  For instance, if only a few greens gain employment, that may be insufficient to incentivize the next generation to acquire skills and high values.  This means that affirmative action policies could have a discontinuous impact based on their size, with no lasting impact until hitting a large enough threshold to induce future investment, and then having a sizeable impact.  

\section{Concluding Remarks}
\label{sec:conclusion}

As discussed by \cite{alfani2021}, inequality has tended to very high levels around the world and throughout history unless actively attacked by policies. We have shown one source of this inequality, and the resulting inefficiency and immobility, is the reliance of the labor market on referrals.  Referral networks, together with naturally-occurring homophily, result in a concentration of job referrals among a privileged group, which, in turn, decreases productivity and mobility and increases inequality.  

As we have seen, this combination of immobility, inequality, and inefficiency can be counteracted by policies that take advantage of network effects---for instance, algorithmic fairness or affirmative action policies---that encourage hiring from the pool.   These policies have long-run impacts even if enacted for short periods. There are also other policies that can have similar effects, for instance, simply subsidizing internships of workers from the pool or providing more information about non-referred workers via various sorts of certification.\footnote{Another policy is to directly counteract the homophily in networks by influencing network formation.  Some of this occurs naturally (\cite{galeottim2014}), but can also be policy influenced, although that is much more challenging (e.g., see the discussion in \cite{jackson2021}).} In addition,  even though full efficiency is not reached in the setting with firing, the ability to fire does improve overall productivity and reduce inequality; which provides insight into one advantage of flexible labor markets.  All of these policies, by increasing current green employment and spreading referrals through the population, have similar long-term network effects. 

Although algorithmic fairness and affirmative action policies have been implemented (in some parts of the world), they are often done on a tiny scale.  For instance, Ivy League universities still have more students from the top 1 percent of the income distribution than the bottom 50 percent \citep{chetty2020income}. Understanding the long-term impact and network leverage of such policies provides stronger arguments for their implementation. Importantly, overcoming persistent inequality and immobility and increasing productivity can ultimately lead to greater growth of the whole economy and Pareto improvements, and the analysis here should make clear that such policies are far from being zero-sum.

While we made a number of simplifying assumptions in our model, our conclusions continue to hold in some natural extensions of our model, and we examine several in the appendix. 

The model that we have presented here should be useful as a base for additional explorations.  For instance, we have not modeled how firms choose how many people to hire.  Firms' decisions on hiring could react to the attractiveness of the pool and how much access they have to referrals,\footnote{For an analysis of the size of firms in reaction to information about labor markets, see \cite{chandrasekhar2020network}.} and so wider-spread referrals could also lead to the creation of new openings, amplifying some of the effects we have shown.  Another effect of referrals is on how long workers last at their current jobs, and allowing for reduced turnover due to referrals would also be interesting to explore but would require changes to the model. One could also explore career choices using this model as a base.  People would (inefficiently) be driven to choose professions of their family and friends, even if they might be more talented in another profession, because of the improved wage and employment prospects due to referral connections. Additionally, firms sometimes hire referred workers as favors to their current workers (or because the current worker lied), even when the referred workers are of low value.  Such effects could exacerbate the advantage to referrals and lead to further increases in inequality and immobility across groups. Finally, and perhaps most importantly, we have taken the network as exogenous.  Given its important implications, it could affect network formation.  Combining our dynamic analysis with a job-network formation model such as that of \cite{galenianos2021referral} would be an important next step.  

{
\small
\bibliographystyle{jpe}
\bibliography{jobnetworks.bib}

@article{alfani2021,
Author = {Alfani, Guido},
Title = {Economic Inequality in Preindustrial Times: Europe and Beyond},
Journal = {Journal of Economic Literature},
Volume = {59},
Number = {1},
Year = {2021},
Pages = {3-44},}

@article{buhai2023social,
  title={A social network analysis of occupational segregation},
  author={Buhai, I Sebastian and van der Leij, Marco J},
  journal={Journal of Economic Dynamics and Control},
  volume={147},
  pages={104593},
  year={2023},
  publisher={Elsevier}
}

@article{okafor2020social,
  title={Social networks as a mechanism for discrimination},
  author={Okafor, Chika O},
  journal={arXiv preprint arXiv:2006.15988},
  year={2020}
}

@article{dhillon2021employee,
  title={Employee Referral, Social Proximity, and Worker Discipline: Theory and Suggestive Evidence from India},
  author={Dhillon, Amrita and Iversen, Vegard and Torsvik, Gaute},
  journal={Economic Development and Cultural Change},
  volume={69},
  number={3},
  pages={1003--1030},
  year={2021},
  publisher={The University of Chicago Press Chicago, IL}
}

@ARTICLE{chettyetal2022I,
title={Social Capital in the United States I: Measurement
and Associations with Economic Mobility},
author={Raj Chetty  and Matthew O. Jackson  and  Theresa Kuchler and Johannes Stroebl   and Nathan Hendren and
Robert Fluegge and Sara Gong and Federico Gonzalez and Armelle Grondin and Matthew Jacob and Drew Johnston and Martin Koenen and
Eduardo Laguna-Muggenburg and  Florian Mudekereza and Tom Rutter and Nicolaj Thor and Wilbur Townsend and Ruby Zhang and Mike Bailey and Pablo Barbera and Monica Bhole and Nils Wernerfelt},
year={2022},
volume={608},
Issue={7921},
pages = {108--121},
journal={Nature}
}

@ARTICLE{chettyetal2022II,
title={Social Capital in the United States II: Determinants
of Economic Connectedness},
author={Raj Chetty  and Matthew O. Jackson  and  Theresa Kuchler and Johannes Stroebl   and Nathan Hendren and
Robert Fluegge and Sara Gong and Federico Gonzalez and Armelle Grondin and Matthew Jacob and Drew Johnston and Martin Koenen and
Eduardo Laguna-Muggenburg and  Florian Mudekereza and Tom Rutter and Nicolaj Thor and Wilbur Townsend and Ruby Zhang and Mike Bailey and Pablo Barbera and Monica Bhole and Nils Wernerfelt},
year={2022},
volume={608},
Issue={7921},
pages={122--134},
journal={Nature}
}

@article{galenianos2021referral,
    author = {Galenianos, Manolis},
    title = "{Referral Networks and Inequality}",
    journal = {The Economic Journal},
    volume = {131},
    number = {633},
    pages = {27--301},
    year = {2020}}

@article{bandiera2009social,
  title={Social connections and incentives in the workplace: Evidence from personnel data},
  author={Bandiera, Oriana and Barankay, Iwan and Rasul, Imran},
  journal={Econometrica},
  volume={77},
  number={4},
  pages={1047--1094},
  year={2009},
  publisher={Wiley Online Library}
}

@article{bandiera2010social,
  title={Social Incentives in the Workplace},
  author={Bandiera, Oriana and Barankay, Iwan and Rasul, Imran},
  journal={The Review of Economic Studies},
  volume={77},
  number={2},
  pages={417--458},
  year={2010},
  publisher={Wiley-Blackwell}
}

@article{bandiera2008social,
  title={Social capital in the workplace: Evidence on its formation and consequences},
  author={Bandiera, Oriana and Barankay, Iwan and Rasul, Imran},
  journal={Labour Economics},
  volume={15},
  number={4},
  pages={724--748},
  year={2008},
  publisher={Elsevier}
}

@article{calvo2005job,
  title={Job matching, social network and word-of-mouth communication},
  author={Calv{\'o}-Armengol, Antoni and Zenou, Yves},
  journal={Journal of Urban Economics},
  volume={57},
  number={3},
  pages={500--522},
  year={2005},
  publisher={Elsevier}
}

@article{chetty2020income,
  title={Income Segregation and Intergenerational Mobility Across Colleges in the United States},
  author={Chetty, Raj and Friedman, John N and Saez, Emmanuel and Turner, Nicholas and Yagan, Danny},
  journal={The Quarterly Journal of Economics},
  volume={135},
  number={3},
  pages={1567--1633},
  year={2020},
  publisher={Oxford University Press}
}

@article{nickell1997,
  title={Unemployment and Labor Market Rigidities: Europe versus North America},
  author={Nickell, Stephen},
  journal={Journal of Economic Perspectives},
  volume={11},
  number={3},
  pages={55--74},
  year={1997}
}

@article{arbex2019network,
  title={Network search: Climbing the job ladder faster},
  author={Arbex, Marcelo and O'Dea, Dennis and Wiczer, David},
  journal={International Economic Review},
  volume={60},
  number={2},
  pages={693--720},
  year={2019},
  publisher={Wiley Online Library}
}

@techreport{chandrasekhar2020network,
  title={Network-Based Hiring: Local Benefits; Global Costs},
  author={Chandrasekhar, Arun G. and Morten, Melanie and Peter, Alessandra},
  year={2020},
  institution={NBER},
  number={26806},
type = "Working Paper"
}

@incollection{fernandez2014causal,
  title={The Causal Status of Social Capital in Labor Markets},
  author={Fernandez, Roberto M. and Galperin, Roman V.},
  booktitle={Contemporary Perspectives on Organizational Social Networks},
 volume={40},
  pages={445--462},
  year={2014},
  publisher={Emerald Group},
    editors = {Daniel J. Brass and Giuseppe Labianca and Ajay Mehra a dDaniel S. Halgin and Stephen P. Borgatti}
}

@article{coate1993will,
  title={Will Affirmative-Action Policies Eliminate Negative Stereotypes?},
  author={Coate, Stephen and Loury, Glenn C.},
  journal={The American Economic Review},
  pages={1220--1240},
  year={1993},
  publisher={JSTOR}
}

@article{benson2019discrimination,
  title={Discrimination in Hiring: Evidence from Retail Sales},
  author={Benson, Alan and Board, Simon and Meyer-ter Vehn, Moritz},
  year={Forthcoming},
  journal={The Review of Economic Studies}
}

@article{dustmann2016,
  title={Referral-based Job Search Networks},
  author={Dustmann, Christian and Glitz, Albrecht and Sch{\"o}nberg, Uta and Br{\"u}cker, Herbert},
  journal={The Review of Economic Studies},
  volume={83},
  number={2},
  pages={514--546},
  year={2016},
  publisher={Oxford University Press}
}

@article{burks2015,
  title={The Value of Hiring through Employee Referrals},
  author={Burks, Stephen V. and Cowgill, Bo and Hoffman, Mitchell and Housman, Michael},
  journal={The Quarterly Journal of Economics},
  volume={130},
  number={2},
  pages={805--839},
  year={2015},
  publisher={MIT Press}
}

@article{beamankm2016,
  title={Do Job Networks Disadvantage Women? Evidence from a Recruitment Experiment in Malawi},
  author={Beaman, Lori A. and Keleher, Niall and Magruder, Jeremy},
  journal={Journal of Labor Economics},
  volume={36},
  number={1},
  pages={121--157},
  year={2018},
  publisher={University of Chicago Press Chicago, IL}
}

@article{zeltzer2020,
  title={Gender Homophily in Referral Networks: Consequences for the Medicare Physician Earnings Gap},
  author={Zeltzer, Dan},
  journal={American Economic Journal: Applied Economics},
  volume={12},
  number={2},
  pages={169--97},
  year={2020}
}

@article{waldman1984,
  title={Job Assignments, Signalling, and Efficiency},
  author={Waldman, Michael},
  journal={The RAND Journal of Economics},
  volume={15},
  number={2},
  pages={255--267},
  year={1984},
  publisher={JSTOR}
}

@article{milgromo1987,
  title={Job Discrimination, Market Forces, and the Invisibility Hypothesis},
  author={Milgrom, Paul and Oster, Sharon},
  journal={The Quarterly Journal of Economics},
  volume={102},
  number={3},
  pages={453--476},
  year={1987},
  publisher={MIT Press}
}

@article{miller2017,
  title={The Persistent Effect of Temporary Affirmative Action},
  author={Miller, Conrad},
  journal={American Economic Journal: Applied Economics},
  volume={9},
  number={3},
  pages={152--90},
  year={2017}
}

@article{holzern2000,
  title={Assessing Affirmative Action},
  author={Holzer, Harry and Neumark, David},
  journal={Journal of Economic Literature},
  volume={38},
  number={3},
  pages={483--568},
  year={2000}
}

@incollection{bondf2019,
  title={Networks for the unemployed?},
  author={Bond, Brittany M. and Fernandez, Roberto M.},
  booktitle={Social Networks at Work},
  pages={276--301},
  year={2019},
  publisher={Routledge}
}

@article{rubineauf2013,
  title={Missing Links: Referrer Behavior and Job Segregation},
  author={Rubineau, Brian and Fernandez, Roberto M.},
  journal={Management Science},
  volume={59},
  number={11},
  pages={2470--2489},
  year={2013},
  publisher={INFORMS}
}

@article{gibbonsk1991,
  title={Layoffs and Lemons},
  author={Gibbons, Robert and Katz, Lawrence F.},
  journal={Journal of Labor Economics},
  volume={9},
  number={4},
  pages={351--380},
  year={1991},
  publisher={University of Chicago Press}
}

@book {Kahneman2011,
	title = {Thinking, Fast and Slow},
	year = {2011},
	publisher = {Macmillan},
	author = {Daniel Kahneman}
}

@book {SenEconomicInequality1973,
	title = {On Economic Inequality},
	year = {1973},
	publisher = {Clarendon Press},
	organization = {Clarendon Press},
	author = {Amartya Sen}
}

@incollection{farber1999,
title = "Mobility and stability: The dynamics of job change in labor markets",
author={Farber, Henry S},
booktitle= {Handbook of Labor Economics},
editor = "Orley C. Ashenfelter and David Card",
volume={3},
  pages={2439--2483},
  year={1999},
  publisher={Elsevier}
}

@techreport{kruegercc2014,
  title={Are the Long-Term Unemployed on the Margins of the Labor Market?},
  author={Krueger, Alan B and Cramer, Judd and Cho, David},
  type={Brookings Papers on Economic Activity},
  pages={229--299},
  year={2014},
  institution={Brookings Institute}
}

@incollection{topa2011,
title = "Labor Markets and Referrals",
booktitle    = "Handbook of Social Economics",
editor = "Jess Benhabib and Alberto Bisin and Matthew O. Jackson",
publisher = "North-Holland",
volume = "1",
pages = "1193--1221",
year = "2011",
author = "Giorgio Topa"
}

@article{patacchiniz2012,
  title={Ethnic networks and employment outcomes},
  author={Patacchini, Eleonora and Zenou, Yves},
  journal={Regional Science and Urban Economics},
  volume={42},
  number={6},
  pages={938--949},
  year={2012},
  publisher={Elsevier}
}

@techreport{chevrot-bianco2021,
 title = "It only takes a strong tie: Board gender quotas and network-based hiring",
 author = "Esther Chevrot-Bianco",
 type = "Conference Paper",
year = "2021",
 institution={IZA Workshop on Gender and Family Economics: Women in Leadership},
}

@article{bohrenhip2019,
    author = {Bohren, J. Aislinn and Haggag, Kareem and Imas, Alex and Pope, Devin G.},
    title = "{Inaccurate Statistical Discrimination: An Identification Problem}",
    journal = {The Review of Economics and Statistics},
    year = {Forthcoming},
}

@techreport{lesterrt2021,
author = {Lester, Benjamin and Rivers, David and Topa, Giorgio},
number = {987},
institution = {Federal Reserve Bank of New York},
title = {The heterogeneous impact of referrals on labor market outcomes},
type = {Staff Report},
year = {2021}
}

@book{jackson2019,
author={Jackson, Matthew O.},
title={The Human Network: How Your Social Position Determines Your Power, Beliefs, and Behaviors},
publisher={Pantheon Books: New York},
year={2019}
}

@article{sacerdote2001,
  title={Peer Effects with Random Assignment: Results for Dartmouth Roommates},
  author={Sacerdote, Bruce},
  journal={The Quarterly Journal of Economics},
  volume={116},
  number={2},
  pages={681--704},
  year={2001}
}

@article{friebelhhz2023,
author = {Friebel, Guido and Heinz, Matthias and Hoffman, Mitchell and Zubanov, Nick},
title = {What Do Employee Referral Programs Do? Measuring the Direct and Overall Effects of a Management Practice},
journal = {Journal of Political Economy},
volume = {131},
number = {3},
pages = {633--686},
year = {2023}
}

@article{kleinberg2017human_machine,
    author = {Kleinberg, Jon and Lakkaraju, Himabindu and Leskovec, Jure and Ludwig, Jens and Mullainathan, Sendhil},
    title = "{Human Decisions and Machine Predictions}",
    journal = {The Quarterly Journal of Economics},
    volume = {133},
    number = {1},
    pages = {237--293},
    year = {2017}
}

@article{quillian2017meta,
  title={Meta-analysis of field experiments shows no change in racial discrimination in hiring over time},
  author={Quillian, Lincoln and Pager, Devah and Hexel, Ole and Midtbøen, Arnfinn H},
  journal={PNAS},
  volume={114},
  number={41},
  pages={10870--10875},
  year={2017},
}

@techreport{lirb2020hiring_as_exploration,
 title = "Hiring as Exploration",
 author = "Li, Danielle and Raymond, Lindsey R and Bergman, Peter",
 institution = "NBER",
 type = "Working Paper",
 number = "27736",
 year = "2020",
}

@techreport{arrowb2004,
  title={Limited Network Connections and the Distribution of Wages},
  author={Arrow, Kenneth J. and Borzekowski, Ron},
  type={Working Paper},
  institution={SSRN},
  year={2004},
  number={632321}
}

@techreport{lalannes2016,
  title={The old boy network: The impact of professional networks on nemuneration in top executive jobs},
  author={Lalanne, Marie and Seabright, Paul},
  year={2016},
type = {SAFE Working Paper},
number= {123}
}

@article{fernandezcm2000,
  title={Social Capital at Work: Networks and Employment at a Phone Center},
  author={Fernandez, Roberto M and Castilla, Emilio J. and Moore, Paul},
  journal={American Journal of Sociology},
  pages={1288--1356},
  year={2000},
 volume = {105},
number={5}
}

@article{heath2018,
author = {Heath, Rachel},
title = {Why Do Firms Hire Using Referrals? Evidence from Bangladeshi Garment Factories},
journal = {Journal of Political Economy},
volume = {126},
number = {4},
pages = {1691--1746},
year = {2018}}

@ARTICLE{brownst2016,
title = {Do Informal Referrals Lead to Better Matches? Evidence from a Firm's Employee Referral System},
author = {Brown, Meta and Setren, Elizabeth and Topa, Giorgio},
year = {2016},
journal = {Journal of Labor Economics},
volume = {34},
number = {1},
pages = {161--209}
}

@article{pallaiss2016,
  title={Why the Referential Treatment: Evidence from Field Experiments on Referrals},
  author={Pallais, Amanda and Sands, Emily Glassberg},
  year={2016},
  volume={124},
    number={6},
  pages={1793--1828},
  journal={Journal of Political Economy}
}

@book{reess1970,
  title={Workers and Wages in an Urban Labor Market},
  author={Rees, Albert and Shultz, George Pratt},
  year={1970},
  publisher={University of Chicago Press}
}

@book{myerss1951,
  title={The Dynamics of a Labor Market},
  author={Myers, Charles A. and Shultz, George P.},
  year={1951},
  publisher={Prentice-Hall}
}

@article{chettyetal2016,
author = {Raj Chetty  and David Grusky  and Maximilian Hell  and Nathaniel Hendren  and Robert Manduca  and Jimmy Narang },
title = {The fading American dream: Trends in absolute income mobility since 1940},
journal = {Science},
volume = {356},
issue = {6336},
pages = {398--406},
year = {2017}}

@techreport{laschever2013,
title={The Doughboys Network: Social Interactions and the Employment of World War I Veterans},
year={2013},
author={Ron A. Laschever},
institution={SSRN},
type={Working Paper},
number={1205543}
}

@inproceedings{raghavanbkl2020evaluating_algorithms,
author = {Raghavan, Manish and Barocas, Solon and Kleinberg, Jon and Levy, Karen},
title = {Mitigating bias in algorithmic hiring: evaluating claims and practices},
year = {2020},
publisher = {Association for Computing Machinery},
booktitle = {Proceedings of the 2020 Conference on Fairness, Accountability, and Transparency},
pages = {469--481},
numpages = {13},
location = {Barcelona, Spain},
series = {FAT* '20}
}

@article{houser2019ai,
  title={Can AI Solve the Diversity Problem in the Tech Industry? Mitigating Noise and Bias in Employment Decision-Making},
  author={Houser, Kimberly},
  journal={Stanford Technology Law Review},
  volume={22},
  year={2019},
  pages={290--354},
}

@techreport{corak2016,
title={Inequality from Generation to Generation:
The United States in Comparison},
author={Miles Corak},
	institution = "IZA",
	type = "{Discussion Paper}",
	number = "9929",
	year = "2016",
}

@article{conde-ruizgp2021,
title = {Statistical discrimination and committees},
journal = {European Economic Review},
volume = {141},
pages = {103994},
year = {2022},
author = {J. Ignacio Conde-Ruiz and Juan José Ganuza and Paola Profeta},
}

@techreport{hederosskp2022,
author={Hederos, Karin and Sandberg, Anna and Kvissberg, Lukas  and Polan, Erik and Admati, Anat R.},
year={2016},
  type={Working Paper},
title={Gender Homophily in Job Referrals: Evidence from a Field Study Among University Students},
institution={SSRN},
number={4103637}}

@article{jackson2021,
	author = {Matthew O. Jackson},
	journal = {Advances in Economics and Econometrics, Theory and Applications: Twelfth World Congress of the Econometric Society, Cambridge University Press},
	title = {Inequality's Economic and Social Roots: the Role of Social Networks and Homophily},
	year = {2024}
}

@ARTICLE{ioannidesd2004,
author = {Ioannides, Yannis M. and  Datcher Loury, Linda},
year = {2004},
title = {Job Information Networks, Neighborhood Effects, and Inequality},
journal = {Journal of Economic Literature},
volume = {424},
number= {4},
pages = {1056--1093}
}

@article{beaman2012,
author = {Beaman, Lori A.},
year = {2012},
title = {Social Networks and the Dynamics of Labour Market Outcomes: Evidence from Refugees Resettled in the U.S.},
journal = {The Review of Economic Studies},
volume = {79 (1)},
pages = {128-161}
}

@article{rees1966,
  title={Information Networks in Labor Markets},
  author={Rees, Albert},
  journal={The American Economic Review},
  volume={56},
  number={1/2},
  pages={559--566},
  year={1966},
  publisher={JSTOR}
}

@Article{mcphersonsc2001,
	author = {McPherson, Miller and Smith-Lovin, Lynn and Cook, James M.},
	title = {Birds of a Feather: Homophily in Social Networks},
	year = 2001,
	volume = 27,
	pages = {415--444},
	journal = {Annual Review of Sociology}
}

@article{galeottim2014,
author = {Galeotti, Andrea and Merlino, Luca Paolo},
title = {Endogenous Job Contact Networks},
journal = {International Economic Review},
volume = {55},
number = {4},
pages = {1201-1226},
year = {2014}
}

@ARTICLE{calvo2004,
  author = {Calv{\'o}-Armengol, Antoni},
  title = {{Job contact networks}},
  journal = {Journal of Economic Theory},
  year = {2004},
  volume = {115},
  pages = {191--206},
  number = {1},
  publisher = {Elsevier}
}

@ARTICLE{calvoj2007,
  author = {Calv{\'o}-Armengol, Antoni and Jackson, Matthew O.},
  title = {{Networks in labor markets: Wage and employment dynamics and inequality}},
  journal = {Journal of Economic Theory},
  year = {2007},
  volume = {132},
  pages = {27--46},
  number = {1},
  publisher = {Elsevier}
}

@ARTICLE{calvoj2004,
  author = {Calv{\'o}-Armengol, Antoni and Jackson, Matthew O.},
  title = {{The Effects of Social Networks on Employment and Inequalityy}},
  journal = {The American Economic Review},
  year = {2004},
  volume = {94},
  pages = {426--454},
  number = {3},
  publisher = {American Economic Association}
}

@ARTICLE{currarinijp2009,
  author = {Currarini, Sergio and Jackson, Matthew O. and Pin, Paolo},
  title = {{An economic Model of Friendship: Homophily, Minorities, and Segregation}},
  journal = {Econometrica},
  year = {2009},
  volume = {77},
  pages = {1003--1045},
  number = {4},
  publisher = {John Wiley \& Sons}
}

@book{granovetter1974,
  title={Getting a Job: A Study of Contracts and Careers},
  author={Granovetter, Mark},
  year={1974},
  publisher={Harvard University Press}
}

@article{RogersonShimerWrigt2005,
Author = {Rogerson, Richard and Shimer, Robert and Wright, Randall},
Title = {Search-Theoretic Models of the Labor Market: A Survey},
Journal = {Journal of Economic Literature},
Volume = {43},
Number = {4},
Year = {2005},
Pages = {959--988}}

@book{granovetter1995,
  title={Getting a Job: A Study of Contracts and Careers},
  author={Granovetter, Mark},
  year={1995},
  edition={2},
  publisher={University of Chicago Press}
}

@ARTICLE{montgomery1991,
  author = {Montgomery, James D.},
  title = {{Social Networks and Labor-Market Outcomes: Toward an Economic Analysis}},
  journal = {The American Economic Review},
  year = {1991},
  volume = {81},
  pages = {1408--1418},
  number = {5},
  publisher = {JSTOR}
}

@ARTICLE{munshi2003,
  author = {Munshi, Kaivan},
  title = {{Networks in the Modern Economy: Mexican Migrants in the US Labor
	Market}},
  journal = {Quarterly Journal of Economics},
  year = {2003},
  volume = {118},
  pages = {549--599},
  number = {2},
  publisher = {MIT Press}
}

@article{phelps1972,
 author = {Edmund S. Phelps},
 journal = {The American Economic Review},
 number = {4},
 pages = {659--661},
 publisher = {American Economic Association},
 title = {The Statistical Theory of Racism and Sexism},
 volume = {62},
 year = {1972}
}

@incollection{fang2011,
title = "Theories of Statistical Discrimination and Affirmative Action: A Survey",
editor = {Jess Benhabib and Alberto Bisin and Matthew O. Jackson},
booktitle    = "Handbook of Social Economics",
publisher = {North-Holland},
volume = {1},
pages = {133-200},
year = {2011},
issn = {1570-6435},
author = {Hanming Fang and Andrea Moro}
}

@article{lundbergs1983,
 author = {Shelly J. Lundberg and Richard Startz},
 journal = {The American Economic Review},
 number = {3},
 pages = {340--347},
 publisher = {American Economic Association},
 title = {Private Discrimination and Social Intervention in Competitive Labor Market},
 volume = {73},
 year = {1983}
}
}

\appendix 

\newpage 
\setcounter{page}{1}
\begin{center}
{\Large {Online Appendix for}}\\[1em]
{\Large {``The Role of Referrals in Immobility, Inequality, and Inefficiency in Labor Markets''}}\\[1.2em]

{\large
{Lukas Bolte, Nicole Immorlica, and Matthew O.\ Jackson}
}
\end{center}

\section{Proofs of Results in Sections \ref{sec:model}, \ref{sec:groups} and \ref{sec:aa} Omitted from the Text}\label{appendix:proofs}

\begin{proof}[Proof of Lemma~\ref{eqw}]
We first argue the existence of the fixed point. Define a correspondence $G$ by
\[
G(\tilde{v}) = \{\max \{\ubar{w},  \E_{\tilde{v},r} [ v_i  |  i \in \pool]\} | {r\in[0,1]}\}.
\]
Varying over $r$ has an impact on the expected value in the pool at atoms in the distribution (that differ from a fixed point). The existence of an equilibrium threshold follows from the fact that $G$ is convex- and compact-valued and upper-hemicontinuous, and is bounded above (by $\E[v_i]$ from Lemma~\ref{lemma:LE}) and below (by $\ubar{w}$). Thus, by Kakutani's Theorem, there is a fixed point. The fact that firms can mix arbitrarily at a fixed point, comes from the fact that if the equilibrium threshold equals the expected value in the pool, $\tilde{v}=\E_{\tilde{v},r}[v_i|i \in \pool]$, the value of the referred worker and the expected value in the pool are then exactly equal, and then keeping or putting such a worker into the pool does not change the average value in the pool. If the equilibrium threshold equals the minimum wage, $\tilde{v}=\ubar{w}$, and the expected value in the pool is strictly less, then adding or removing the referred worker does not change the fact that the expected value in the pool is strictly less than the equilibrium threshold.  

The uniqueness of the fixed point is shown as follows. Suppose to the contrary that there are two distinct equilibrium thresholds, ${\tilde{v}}<{\tilde{v}}'$. We provide the argument for the case in which ${\tilde{v}}'> \ubar{w}$ as otherwise it must be that ${\tilde{v}}={\tilde{v}}'=\ubar{w}$, since the correspondence is lower-bounded by $\ubar{w}$. The pool for threshold ${\tilde{v}}'$ consists of workers without any referral, referred workers with values not exceeding ${\tilde{v}}$, and referred workers with value between ${\tilde{v}}$ and ${\tilde{v}}'$ (and possibly some at exactly ${\tilde{v}}'$ depending on the mixing). The expected productivity of workers without referrals and referred workers with values not exceeding ${\tilde{v}}$ is at most ${\tilde{v}}$, since ${\tilde{v}}$ is an equilibrium threshold. Averaging the value of this group with the expected value of workers with referral and values between ${\tilde{v}}$ and ${\tilde{v}}'$ gives the expected productivity in the pool if the threshold is ${\tilde{v}}'$; but this cannot result in an expected value of ${\tilde{v}}'$ as ${\tilde{v}}<{\tilde{v}}'$.

The equilibrium behavior of firms and the wages of workers are immediate.
\end{proof} 

\begin{proof}[Proof of Lemma~\ref{lemma:LE}]
Clearly $\E_{\tilde{v},r}[v_i|i \in \pool] \leq \E[v_i]$, since the pool is a convex combination of workers who have no referrals, who thus have the unconditional expected value, and workers who had referrals but had a value no higher than a threshold, and thus have expected values no higher than the unconditional expected value. If the threshold is equal to the minimum value of $F$, then the expected value in the pool is as well and we are done. Thus, since $P(0)<1$, there are some referrals, and since $F$ is nondegenerate it follows that there is a positive mass of workers with values above and a positive mass below the threshold (which is at most the expected value and strictly greater that the minimum value), and thus the expected value of rejected referred workers is strictly below the unconditional expected value, which then means that the expected value of the pool is also below.
\end{proof}

\begin{proof}[Proof of Proposition~\ref{proposition:group_outcomes}]
Note that referral imbalance and employment bias imply that
\begin{multline*}
\frac{n_b}{n_g}\leq \frac{R_b}{R_g} 
= \frac{h_b n_b+(1-h_g) n_g}{h_gn_g + (1-h_b)n_b}
= \frac{h_b \frac{n_b}{n}+(1-h_g) \frac{n_g}{n}}{h_g \frac{n_g}{n} + (1-h_b)\frac{n_b}{n}}
\leq  \frac{h_b {e_b}+(1-h_g) {e_g}}{h_g e_g + (1-h_b)e_b}.
\end{multline*}
where the last inequality also uses the assumption that $h_b\geq 1-h_g$. The last expression gives the ratio of the masses of referrals each group gets. The inequality implies that blue workers get more referrals per capita than green workers.

The first claimed comparative static now follows from the fact that the wage distribution of a worker with at least $2$ referrals first-order stochastically dominates the wage distribution of a worker with exactly $1$ referral which first-order stochastically dominates the wage distribution of a worker with no referrals; and, as the distribution of referrals for blue workers first-order stochastically dominates that of green workers as blue workers get more referrals per capita. The second comparative static also follows as having more referrals per capita implies a higher the employment rate. As the inequality above becomes strict if either of the conditions is strict, so do the first-order dominance comparison of the wage distributions and the difference in employment rates. (Recall that we assumed $\hat{P}(0|m)$ to be strictly decreasing in $m$.)

The third claim---the wage and employment distribution of blue workers dominating those of green workers in all future periods---follows immediately, by simply repeatedly applying the first two claims (since by the second claim, an employment bias favoring blues will always ensue).

For the final claim, we make use of a couple of results stated in Section~\ref{sec:Pchanges} of the supplemental appendix. By Lemma~\ref{lem:concentration}, $P(0)$ is strictly higher in the equilibrium starting with $(e_b,e_g)$ than in the equilibrium starting with $(e_b',e_g')$, and so the conclusion follows from Proposition~\ref{lemma:inequality_and_productivity}.
\end{proof}

We also prove the claims following the proposition about average productivities. First, we show that the average productivity of employed green workers is lower than that of employed blue workers.  Workers with value less than $\tilde{v}$ have the same probability of being hired regardless of their group as they are hired from the pool.  Workers with value more than $\tilde{v}$ are hired on the referral market if they have at least one referral.  As blues have (weakly) more referrals than greens and thus their distribution of referrals first-order stochastically dominates that of green workers,  blue workers with value more than $\tilde{v}$ are (weakly) more likely to be hired overall than green workers with value more than $\tilde{v}$ (and strictly so if either of the conditions is strict).  Together, these observations show the desired comparison of the productivity of employed green versus blue workers. 

Finally, we show that the average productivity of unemployed green workers is at least as high as that of blues.  Note that the productivity of an unemployed green (resp.\ blue) worker is equal to the productivity of a green (resp.\ blue) worker in the pool.  Let
$$
\E^g_{\tilde{v},r}[v_i|i\in \pool\cap\text{green}] \coloneqq   \frac{ P_g(0) \E[ v_i  ]  +  (1-P_g(0)) (\Pr(v_i<{\tilde{v}})\E [v_i  |  v_i< {\tilde{v}} ]+\Pr(v_i=\tilde{v} )(1-r)\tilde{v}) }{P_g(0) +(1-P_g(0)) (\Pr(v_i<{\tilde{v}})+\Pr(v_i=\tilde{v})(1-r))};
$$
$$\E^b_{\tilde{v},r}[v_i|i\in \pool\cap\text{blue}] \coloneqq   \frac{ P_b(0) \E[ v_i  ]  +  (1-P_b(0)) (\Pr(v_i<{\tilde{v}})\E [v_i  |  v_i< {\tilde{v}} ]+\Pr(v_i=\tilde{v} )(1-r)\tilde{v}) }{P_b(0) +(1-P_b(0)) (\Pr(v_i<{\tilde{v}})+\Pr(v_i=\tilde{v})(1-r))}.\\
$$
be the productivity in the pool of greens and blues, where $P_g(0)$ (resp.\ $P_b(0)$) is the probability that a green (resp.\ blue) worker gets no referral. Since $P_g(0)\geq P_b(0)$ as the referral distribution of blue workers first-order stochastically dominates that of green workers, it follows that $\E^g_{\tilde{v},r}[v_i|i\in \pool\cap\text{green}]\geq \E^b_{\tilde{v},r}[v_i|i\in \pool\cap\text{blue}]$, proving the claim.

\begin{proof}[Proof of Lemma~\ref{lem:longrun}]
We first prove that there exists a unique steady state of the employment rates. Let $P_b^t(0)$ and $P_g^t(0)$ be the probabilities of not getting a referral for blues and greens if employment levels are given by $e_b^t$ and $e_g^t$, respectively. Then the difference in employment rates $\Delta \coloneqq \frac{e_b^{t+1}}{n_b}-\frac{e_g^{t+1}}{n_g}$, in period $t+1$ is 
$$
\Delta=\left(P_b^t(0)\cdot Q^t+(1-P_b^t(0))\cdot R^t\right)
-
\left(P_g^t(0)\cdot Q^t+(1-P_g^t(0))\cdot
R^t\right)
$$
where $Q^t$ is the probability of being hired in the pool and 
$$
R^t=\left(\Pr(v_i>\tilde{v}^t)+r\Pr(v_i=\tilde{v}^t)+Q^t(1-r)\Pr(v_i=\tilde{v}^t)+Q^t\Pr(v_i<\tilde{v}^t)\right)
$$
is the probability of being hired given that you got a referral.  Simplifying, we see
\begin{equation}\label{eq:proof_longrun}
\Delta
= \left(P_b^t(0)-P_g^t(0)\right)\cdot \left(Q^t-R^t\right).
\end{equation}
First, note that $P_b^t(0)$ and $P_g^t(0)$ are continuous in $e_b^t$ and $e_g^t$. To see this, note that $\hat{P}(k|m)$ is continuous in $m$ for all $k$ as $m$ is the mean of the distribution and since $\hat{P}(\cdot|m')$ first-order stochastically dominates $\hat{P}(\cdot|m)$ for $m'>m$. It then follows that the right-hand side of \eqref{eq:proof_longrun} gives a convex- and compact-valued, upper hemicontinuous correspondence with respect to period $t$ employment masses if one varies $r$ whenever an atom at the referral threshold exists. Furthermore, if $e^t_b=\min\{1,n_b\}$, then mechanically it must be that $e_b^{t+1}\leq e_b^t$ (and thus $\Delta \leq \frac{e_b^{t}}{n_b}-\frac{e_g^{t}}{n_g}$). Similarly, if $e^t_g=\min\{1,n_g\}$, then $e_g^{t+1}\leq e^t_g$ (and thus $\Delta\geq  \frac{e_b^{t}}{n_b}-\frac{e_g^{t}}{n_g}$). Thus, a fixed point of employment rates (or equivalently employment or referral masses) exists. By assumption, the induced expected value in the pool is above the minimum wage so that firms indeed hire from the pool.

Next, we show that the fixed point is unique. If $h_b=1-h_g$, then referrals per capita are always equal across groups and the unique fixed point is characterized by equal employment rates. Thus, assume that $h_b>1-h_g$ for the remainder of the proof. 

We make use of the following lemma which is proved after this proof.
\begin{lemma}\label{lem:extreme}
Let $\overline{e}_b$ and $\overline{e}_g$ with $\overline{e}_g=1-\overline{e}_b$ uniquely solve
$$
\frac{1}{n_b}(h_b\overline{e}_b+(1-h_g)\overline{e}_g)= \frac{1}{n_g}(h_g\overline{e}_g+(1-h_b)\overline{e}_b),
$$
i.e., the employment masses at which the two groups get the same referrals per capita.

Take any two employment levels in period $t$, $e_b^t$ and $\tilde{e}_b^{t}$, and let the next-period employment levels be denoted by $e_b^{t+1}$ and $\tilde{e}_b^{t+1}$ respectively; with $e_g^t=1-e_b^t, \tilde{e}_g^t=1-\tilde{e}_b^t,e_g^{t+1}=1-e_b^{t+1}, \tilde{e}_g^{t+1}=1-\tilde{e}_b^{t+1}$. If $\tilde{e}_b^{t}>e_b^t\geq \overline{e}_b$, then 
\begin{equation*}
    \left(\frac{\tilde{e}_b^{t+1}}{n_b}-\frac{\tilde{e}_g^{t+1}}{n_g}\right) - \left(\frac{e_b^{t+1}}{n_b}-\frac{e_g^{t+1}}{n_g}\right)<\left(\frac{\tilde{e}_b^{t}}{n_b}-\frac{\tilde{e}_g^{t}}{n_g}\right) - \left(\frac{e_b^{t}}{n_b}-\frac{e_g^{t}}{n_g}\right).
\end{equation*}
Similarly, if $\tilde{e}_b^{t}<e_b^t\leq \overline{e}_b$, then
\begin{equation*}
    \left(\frac{\tilde{e}_b^{t+1}}{n_b}-\frac{\tilde{e}_g^{t+1}}{n_g}\right) - \left(\frac{e_b^{t+1}}{n_b}-\frac{e_g^{t+1}}{n_g}\right)>\left(\frac{\tilde{e}_b^{t}}{n_b}-\frac{\tilde{e}_g^{t}}{n_g}\right) - \left(\frac{e_b^{t}}{n_b}-\frac{e_g^{t}}{n_g}\right).
\end{equation*}
\end{lemma}

Let $e_b$ be a fixed point and consider any employment level $\widehat{e}_b^t\neq e_b$; with $\widehat{e}_g^t=1-\widehat{e}_b^t$. If $e_b,\widehat{e}_b^t\geq \overline{e}_b$ or $e_b,\widehat{e}_b^t\leq \overline{e}_b$ , then Lemma~\ref{lem:extreme} directly implies that $\widehat{e}_b^t$ is not a fixed point. To see this, e.g., consider the case $e_b>\widehat{e}_b^t\geq \overline{e}_b$. By Lemma~\ref{lem:extreme}, letting $e_b^t=\widehat{e}_b^t$ and $\tilde{e}_b^t=e_b$, we have 
$$ 
    \left(\frac{e_b}{n_b}-\frac{e_g}{n_g}\right) - \left(\frac{\widehat{e}_b^{t+1}}{n_b}-\frac{\widehat{e}_g^{t+1}}{n_g}\right)<\left(\frac{e_b}{n_b}-\frac{{e}_g}{n_g}\right) - \left(\frac{\widehat{e}_b^{t}}{n_b}-\frac{\widehat{e}_g^{t}}{n_g}\right),
$$
as $e_b$ is a fixed point so that $\widehat{e}_b^t$ is not a fixed point.

Suppose now that $e_b<\overline{e}_b\leq \widehat{e}_b^t$. By Lemma~\ref{lem:extreme} letting $\tilde{e}_b^t=e_b$ and $e_b^t=\overline{e}_b$, we have 
$$
    \left(\frac{{e}_b}{n_b}-\frac{e_g}{n_g}\right) - 0  >\left(\frac{{e}_b}{n_b}-\frac{{e}_g}{n_g}\right) - \left(\frac{\overline{e}_b}{n_b}-\frac{\overline{e}_g}{n_g}\right),
$$
as the two groups get the same referrals per capita when the employment of blue workers is given by $\overline{e}_b$, and again as $e_b$ is a fixed point. Again by Lemma~\ref{lem:extreme}, now letting $e_b^t=\overline{e}_b$ and $\tilde{e}_b^t = \widehat{e}_b^t$, we have
$$
    \left(\frac{\widehat{e}_b^{t+1}}{n_b}-\frac{\widehat{e}_g^{t+1}}{n_g}\right) -0\leq \left(\frac{\widehat{e}_b^{t}}{n_b}-\frac{\widehat{e}_g^{t}}{n_g}\right) - \left(\frac{\overline{e}_b}{n_b}-\frac{\overline{e}_g}{n_g}\right),
$$
where equality holds if $\widehat{e}_b^t=\overline{e}_b$. Combining the two previous inequality yields
\begin{equation*}
    \frac{\widehat{e}_b^{t+1}}{n_b}-\frac{\widehat{e}_g^{t+1}}{n_g}<\frac{\widehat{e}_b^{t}}{n_b}-\frac{\widehat{e}_g^{t}}{n_g},
\end{equation*}
so that $\widehat{e}_b^t$ is not a fixed point.

Lastly, the final case, $\widehat{e}_b^t\leq \overline{e}_b<e_b$, is proven analogously to the previous case by interchanging the subscripts as the role of green and blue workers is symmetric, and is thus omitted.

Let us now show that referral imbalance goes hand in hand with long-run bias in employment rates. As before, let $e_b,e_g$ be be the fix point; further, let $R_b(e_b,e_g)$ denote the associated mass of referrals blues get; with $e_g=1-e_b$ and $R_g(e_b,e_g)$ defined analogously. 
We have
\begin{multline*}
	\frac{e_b}{e_g} = \frac{n_b}{n_g} 
	\iff \frac{R_b(e_b,e_g)}{R_g(e_b,e_g)} = \frac{n_b}{n_g} 
	\iff \frac{e_bh_b + e_g(1-h_g)}{e_gh_g + e_b(1-h_b)}=\frac{n_b}{n_g} 
	\iff  \frac{R_b}{R_g} = \frac{n_b}{n_g}.
\end{multline*}
As employment and referral are positively linked, we can further say that $\frac{e_b}{e_g}\geq \frac{n_b}{n_g} \iff \frac{R_b}{R_g}\geq \frac{n_b}{n_g}$.

Finally, we show that the system of employment rates converges to this fixed point if there is no referral imbalance, i.e., $\frac{R_b}{R_g}=\frac{n_b}{n_g}$. As there is no referral imbalance, the unique fixed point features equal employment rates.

First, by assumption, the expected value in the pool induced by the steady state of equal employment rates is above the minimum wage. As with referral balance, the probability no referrals is minimized with equal employment rates, so is the expected value in the pool. Hence, for any employment levels, firms will be hiring from the pool.

Consider any initial employment, $e_b^t,e_g^t$, and suppose without loss that $\frac{e_b^t}{e_g^t}>\frac{n_b}{n_g}$. Then
\begin{equation*}
    h_be_b^t+(1-h_g)e_g^t>R_b, \quad h_ge_g^t+(1-h_b)e_b^t<R_g 
\end{equation*}
implying
\begin{equation*}
    \frac{h_be_b^t+(1-h_g)e_g^t}{h_ge_g^t+(1-h_b)e_b^t}>\frac{R_b}{R_g}=\frac{n_b}{n_g}.
\end{equation*}
Therefore blue workers get relatively more referrals per capita than green workers. Consequently, blue workers have a lower probability of not getting a referral and thus have a higher employment rate at the end of the period than greens. By Lemma~\ref{lem:extreme}, letting $e_b^t=\overline{e}_b=\frac{n_b}{n}$ and $\tilde{e}_b^t=e_b^t$, we have 
\begin{equation*}
    \frac{e_b^{t+1}}{n_b}-\frac{e_g^{t+1}}{n_g}<\frac{e_b^{t}}{n_b}-\frac{e_g^{t}}{n_g},
\end{equation*}
i.e., the employment of blue workers is decreasing whenever $\frac{e_b^t}{e_g^t}>\frac{e_b}{e_g}$. As the employment of blue workers is bounded below, because per capita it is always greater than that of green workers, it must converge, and therefore converges to the unique fixed point.
\end{proof}

\begin{proof}[Proof of Lemma~\ref{lem:extreme}]
Suppose $\tilde{e}_b^{t}>e_b^t\geq \overline{e}_b$. Consider the first factor of \eqref{eq:proof_longrun}. We will use two facts to eventually bound this factor. First, the difference in the average number of referrals a blue worker gets given the employment levels $\tilde{e}_b^t,\tilde{e}_g^t$ and $e_b^t,e_g^t$ is given by 
\begin{align*}
    \frac{1}{n_b}{h_b\tilde{e}_b^t+h_g\tilde{e}_g^t}-\frac{1}{n_b}{h_b{e}_b^t+h_g{e}_g^t} &= \frac{1}{n_b}{h_b\tilde{e}_b^t+h_g(1-\tilde{e}_b^t)}-\frac{1}{n_b}{h_b{e}_b^t+h_g(1-{e}_b^t)} \\
    &= \frac{1}{n_b}(\tilde{e}_b^t-e_b^t)(h_b-(1-h_g)) \\
    &< \frac{\tilde{e}_b^t}{n_b}-\frac{e_b^t}{n_b},
\end{align*}
where the last step follows as $\tilde{e}_b^t>e_b^t$ and $h_b>1-h_g$. Second, as this difference is positive, $\tilde{P}_b^t$ first-order stochastically dominates $P_b^t$. We use these two facts to bound the changes in the probability of not getting a referral for blue and green workers. As $\tilde{P}_b^t$ first-order stochastically dominates $P_b^t$, for all $k'$
$$
    E_b^t[k|k\geq k',\tilde{e}_b^t,\tilde{e}_g^t]\geq E_b^t[k|k\geq k',{e}_b^t,{e}_g^t].
$$
As 
\begin{align*}
       E_b^t[k|k\geq 1,e_b,{e}_g](1-P_b(0|e_b,e_g))&=E_b[k]
\end{align*}
we can express and bound difference in the probability of not getting a referral for blue workers as 
\begin{align*}
    \tilde{P}_b^t(0)-P_b^t(0)&=-\left(\frac{E_b[k|\tilde{e}_b^t,\tilde{e}_g^t]}{E_b[k|\tilde{e}_b^t,\tilde{e}_g^t,k\geq 1]} - \frac{E_b[k|{e}_b^t,{e}_g^t]}{E_b[k|{e}_b^t,{e}_g^t,k\geq 1]}\right) \\
    &\geq -\left(\frac{E_b[k|\tilde{e}_b^t,\tilde{e}_g^t]}{E_b[k|\tilde{e}_b^t,\tilde{e}_g^t,k\geq 1]} - \frac{E_b[k|{e}_b^t,{e}_g^t]}{E_b[k|\tilde{e}_b^t,\tilde{e}_g^t,k\geq 1]}\right) \\
    &\geq -\left({E_b[k|\tilde{e}_b^t,\tilde{e}_g^t]} - {E_b[k|{e}_b^t,{e}_g^t]}\right) \\
    & > -\left(\frac{\tilde{e}_b^t}{n_b}-\frac{e_b^t}{n_b}\right).
\end{align*}

One can analogously show that 
$$
    \tilde{P}_g^t(0)-P_g^t(0)<\left(\frac{\tilde{e}_b^t}{n_b}-\frac{e_b^t}{n_b}\right).
$$

Consider the second factor of \eqref{eq:proof_longrun}. By Lemma \ref{lem:concentration}, the aggregate probability of not getting a referral is higher given employment levels $\tilde{e}_b^t,\tilde{e}_g^t$ than $e_b^t,e_g^t$.

As a result, fewer workers are vetted, the lemons effect decreases and the equilibrium threshold increases (Proposition~\ref{lemma:inequality_and_productivity}). Both the increase in the probability of not getting a referral and in the equilibrium threshold imply that fewer workers are hired on the referral market given $\tilde{e}_b^t$ than given $e_b^t$; furthermore implying that $\tilde{Q}^t>Q^t$, i.e., the probability of being hired from the pool is larger given $\tilde{e}_b^t$ than given $e_b^t$. 
$$ 
    Q^t-R^t = -\left(1-(1-r)\Pr(v_i=\tilde{v}^t)-\Pr(v_i<\tilde{v}^t)\right)(1-Q^t),
$$
and similarly for $\tilde{Q}^t-\tilde{R}^t$, it must be that the second factor of \eqref{eq:proof_longrun}, when moving from current employment $e_b^t$ to $\tilde{e}_b^t$, increases while still remaining negative. That is, $Q^t-R^t<\tilde{Q}^t-\tilde{R}^t\leq 0$. Let $P_b^t$ and $\tilde{P}_b^t$ denote $P^t_b(0|e_b^t,e_g^t)$ and $P^t_b(0|\tilde{e}_b^t,\tilde{e}_g^t)$; with analogous notation for green workers. Putting this together, we see
\begin{align*}
    \left(\frac{\tilde{e}_b^{t+1}}{n_b}-\frac{\tilde{e}_g^{t+1}}{n_g}\right)-\left(\frac{e_b^t}{n_b}-\frac{e_g^t}{n_g}\right) &= (\tilde{P}_b^t-\tilde{P}_g^t)(\tilde{Q}^t-\tilde{R}^t)-(P_b^t-P_g^t)(Q^t-R^t)\\
    &\leq (\tilde{P}_b^t-\tilde{P}_g^t)(Q^t-R^t)-(P_b^t-P_g^t)(Q^t-R^t)\\
    &=-(\tilde{P}_b^t-P_b^t)(-(Q^t-R^t)) +(\tilde{P}_g^t-P_g^t))(-(Q^t-R^t))\\
    &<\left(\frac{\tilde{e}_b^t}{n_b}-\frac{e_b^t}{n_b}\right)-\left(\frac{\tilde{e}_g^t}{n_g}-\frac{e_g^t}{n_g}\right)\\
    &= \left(\frac{\tilde{e}_b^t}{n_b}-\frac{\tilde{e}_g^t}{n_g}\right)-\left(\frac{e_b^t}{n_b}-\frac{e_g^t}{n_g}\right),
\end{align*}
as required.

Suppose now that $\tilde{e}_b^t<e_b^t \leq \overline{e}_b$. Then $\tilde{e}_g^t>e_g^t \geq \overline{e}_g$. Interchanging the subscripts and applying the previous part of this proof gives 
\begin{equation*}
        \left(\frac{\tilde{e}_g^{t+1}}{n_g}-\frac{\tilde{e}_b^{t+1}}{n_b}\right) - \left(\frac{e_g^{t+1}}{n_g}-\frac{e_b^{t+1}}{n_b}\right)<\left(\frac{\tilde{e}_g^{t}}{n_g}-\frac{\tilde{e}_b^{t}}{n_b}\right) - \left(\frac{e_g^{t}}{n_g}-\frac{e_b^{t}}{n_b}\right),
\end{equation*}
which is a simple rearrangement of the claimed inequality.
\end{proof}

\begin{proof}[Proof of Proposition~\ref{prop:fair}]
We first prove the claimed ranking in total production. Analogous arguments to those in the proof of Lemma~\ref{eqw} show that the value in the pool as a function of the hiring threshold is a bounded upper hemicontinuous correspondence, and the equilibrium hiring threshold is unique. For any given threshold (and associated composition of workers in the pool), a firm's value of hiring from the pool is highest when it can target unreferred workers, followed by when it can target groups, and lowest if it cannot target either. It follows that equilibrium hiring thresholds follow this ranking.

The ranking of hiring thresholds translates into a ranking of total production. Let $\tilde{v}$ be an equilibrium threshold, either with or without targeting. Since $\tilde{v}$ is the expected production of firms hiring from the pool as we assume full hiring from the pool, total production is given by
$$
n(1-P(0))E[\max\{v_i,\tilde{v}\}]+(1-n(1-P(0)))\tilde{v}+(1-n)\ubar{w}.
$$
The above expression increases in $\tilde{v}$, and so the claim follows.

We next prove that green employment increases under either targeting regime. Under either targeting regime, the hiring threshold increases. Thus, fewer blue workers are hired on the referral market. This directly implies that green employment increases if workers are targeted by group (in this case, additional blue workers are only hired if all green workers are hired). This also implies that green employment increases if workers are targeted by referral status because the now rejected referred workers are disproportionately blue (relative to population shares), and more firms now hire from the pool, which favors greens. (Note that hiring from the pool favors greens for two reasons: the pool is disproportionately green, and relatively more greens enter the pool without referral.)
\end{proof}

\begin{proof}[Proof of Proposition~\ref{aamultiplier}]
Recall that $h_g=h_b=1$ and the initial employment in period one, i.e., $e_g^0\coloneqq e_g$, is biased towards blues (i.e., $\frac{e_b^0}{e_g^0}> \frac{n_b}{n_g}$). A small increase in $e_g^0$ will not change the bias in employment. Thus, by Proposition~\ref{proposition:group_outcomes} the initial employment in any period $t$ is biased with and without the increase in $e_g^0$ (i.e., $\frac{e_b^{t-1}}{e_g^{t-1}}<\frac{n_b}{n_g}$ for all finite $t\geq1$).  

Therefore, if green employment strictly increases in all future periods, this implies strict increases in total production as well by Proposition~\ref{lemma:inequality_and_productivity}. The remainder of this proof is thus devoted to showing there is a strict increase in green employment in all future periods.

The proof consists of three steps.  We first show that for all but finitely many levels of initial employment $e_g^{t-1}$ of green workers, final employment $e_g^t$ is strictly increasing and continuous in $e_g^{t-1}$.  We then use this to argue that, over a finite number of periods $1,\ldots,T$, for all but finitely many initial employment levels $e_g^0$, final employment $e_g^t$ is strictly increasing and continuous in $e_g^0$.  Finally, we argue that we can choose $T$ large enough that the convergence in employment levels from period $T$ onward is monotone.

\paragraph{Step 1:} 
We first show that for all but finitely many initial employment levels  $e_g^{t-1}$ of green workers in period $t$, final employment $e_g^t$ is strictly increasing and continuous in $e_g^{t-1}$. Let $\tilde{v}(e_g)$ denote the equilibrium threshold if initial green employment is given by $e_g$. 

We proceed in cases. First, suppose that the employment level is such that firms weakly prefer not to hire from the pool, i.e., employment levels in the set $B\coloneqq \{e_g:\tilde{v}(e_g)=\ubar{w}\}$. If firms do not hire from the pool, then employment $e_g^t$ is strictly increasing and continuous in $e_g^{t-1}$ as, in this case,
$$ 
    e_g^t= (1-P_g(0|1-e_g^{t-1},e_g^{t-1}))n_g(\Pr(v_i>\ubar{w})+r\Pr(v_i=\ubar{w}))
$$
and $P_g(0|1-e_g,e_g)$ is strictly decreasing and continuous in $e_g$; and $r$ is constant ($r=1$) by assumption. Continuity of $P_g(0|1-e_g,e_g)$ is implied by the assumed first-order stochastic dominance ordering. If firms are indifferent about hiring from the pool, then the choice of mixing parameter $r$ for hiring on the referral market does not impact the expected value in the pool.  Furthermore, the expected value in the pool must be $\ubar{w}$. Therefore, we can solve Equation~\eqref{tau} for $P(0)$.  As $P(0)$ is strictly decreasing and continuous in $e_g$, this implies  there is only one employment level $e_g^{t-1}$ that leads to firms being indifferent about hiring from the pool (in which case $e_g^t$ may decrease).

Next suppose that the employment level $e_g^{t-1}$ is such that firms prefer to hire from the pool, i.e.,  $e_g^{t-1}\in C\coloneqq \{e_g:\tilde{v}(e_g)>\ubar{w}\}$. Note that $B$ and $C$ form a partition of all possible employment levels. Let $A\coloneqq \{e_g:\tilde{v}(e_g) \in \supp (F)\}$ and let $D \coloneqq C \cap A$. We use the following fact.
\begin{fact}\label{fact1}
For $e_g\in C$, $\tilde{v}(e_g)$ is strictly decreasing in $e_g$.
\end{fact}
\noindent To see this, first note for $e_g\in C$, $\tilde{v}(e_g)=E_{\tilde{v},r}[v_i|i \in \pool]$. Furthermore, for such $e_g$, the choice of mixing parameter $r$ does not impact the expected value in the pool. Thus by Proposition~\ref{lemma:inequality_and_productivity}, $\tilde{v}(e_g)$ is decreasing in $P(0)$. Furthermore, as before, $P(0)$ is strictly decreasing and continuous in $e_g$ (i.e., as employment becomes more equal, and given that workers refer their own group, fewer workers have no referrals), implying the fact. 

Fact~\ref{fact1}, combined with the assumption that $F$ has finite support, implies $|D|<\infty$. $D$ can thus be enumerated as $D=\{x_0,x_1,\dots,x_N\}$ with  $x_i<x_{i+1}$ for $i=0,\dots,N-1$. We claim that employment $e_g^t$ is strictly increasing and continuous in $e_g^{t-1}$ on each open interval $(x_i,x_{i+1})$ for $i=0,1,\dots, N-1$, and on $[0,x_0)$ and $(x_N, \min\{1,n_g\}]$. To see this, note that $\E_{\tilde{v},r}[v_i|i \in \pool]$ is constant in $\tilde{v}$ for $\tilde{v}\not\in \supp (F)$. Furthermore, as an increase in $e_g^{t-1}$ continuously decreases $P(0)$, and $e_g^{t-1} \in (x_i,x_{i+1})$ implies $\tilde{v}(e_g^{t-1})\not\in\supp (F)$, increasing $e_g^{t-1}$ continuously decreases $\E_{\tilde{v},r}[v_i|i \in \pool]$ so that the equilibrium threshold must also changes continuously. As a result, we may ignore the change in the equilibrium threshold when assessing the induced change in $e_g^{t}$. Thus we can write $e_g^t$ as a function of the probabilities of blues and greens being hired on the referral market and the hiring threshold $\tilde{v}(e_g^{t-1})$: 
$$
    e_g^t=n_g\left(M_g^t+(1-M_g^t)Q^t\right),
$$
where $M_g^t=(1-P_g(0))(\Pr(v_i>\tilde{v}(e_g^{t-1})))$ is the probability a green worker is hired on the referral market (resp.\ $M_b^t$) and $Q^t=(1-(M_g^t+M_b^t))/(n-(M_g^t+M_b^t))$ is the probability a worker is hired in the pool. As $P_g(0)$ strictly decreases continuously and $P_b(0)$ strictly increases continuously, and $e_g^t$ is a continuous function of $P_g(0)$ and $P_b(0)$, it must be that $e_g^t$ indeed strictly increases and is continuous. Thus, for any time period $t$, for all but finitely many values of the initial employment, $e_g^{t-1}$, the resulting employment, $e_g^t$, is strictly increasing and continuous in $e_g^{t-1}$. 

\paragraph{Step 2:} 
We next show for all $T<\infty$, there exists a finite set of initial employment levels $E^0$ such that, if $e_g^0\not\in E^0$, then $e_g^T$ is continuous and strictly increasing in $e_g^0$. 

Let $E$ denote the values of $e_g^{t-1}$ for which $e_g^t$ is not strictly increasing and continuous in $e_g^{t-1}$ at $e_g^{t-1}$. By Step 1, $|E|<\infty$. Let $\rho(e_g)$ denote the final green employment when the initial employment is $e_g$. Set $E^T=\emptyset$ and for $t=1,\dots, T$ define
$$  
    E^{t-1}\coloneqq E \cup \rho^{-1}(E^t).
$$
By Step 1, the preimage for any finite set of values of $e_g^{t}$ is finite, as final employment is strictly increasing in initial employment on each interval in $D$ and there are finitely many intervals. Thus, as $E$ is finite, all sets $E^{t}$ for $t=0,\dots,T$ are finite; in particular, $E^0$ is finite. We argue by induction that $e_g^t=\rho^t(e_g)$ is continuous and strictly increasing in $e_g$ at all points $e_g\not\in E^0$ (and hence at all but finitely many points). Suppose $e_g^{t-1}$ is continuous and strictly increasing in $e_g^0$.  Note by construction $e_g^{t-1}\not\in E^{t-1}$.  As $E\subset E^{t-1}$, and $e_g^t=\rho(e_g^{t-1})$, Step 1 implies $e_g^t$ is continuous and increasing at $e_g^{t-1}$ and hence, by the inductive hypothesis, at $e_g^0$.

\paragraph{Step 3:} We now extend the previous claim to the infinite. Note that by our assumption that $h_b=h_g=1$, there is referral balance. Therefore, by Lemma~\ref{lem:longrun}, the employment rates converge. This implies that $P(0)$ converges and as $\tilde{v}$ is monotone in $P(0)$ (by Proposition~\ref{lemma:inequality_and_productivity}) and bounded, so does $\tilde{v}$. Again by Lemma~\ref{lem:longrun}, there is a unique steady state, and for all $e_g^0$, employment rates converge so that the limit must be the unique steady state. With referral balance, this steady state must feature proportional employment masses, i.e., $\frac{e_b}{e_g}=\frac{n_b}{n_g}$. The limit point of the equilibrium threshold is thus given by $\tilde{v}(\frac{n_g}{n})$. As the equilibrium threshold converges, there exists some $T<\infty$, so that $[\tilde{v}(e_g^T),\tilde{v}(\frac{n_g}{n}))\cap \supp (F)= \emptyset$. Then, $e_g^{t}$ is strictly increasing in $e_g^{t-1}$ for all $t\geq T$. Using the previous steps, we know that there exists a small enough increase in $e_g^0$ so that $e_g^t$ increases in all periods up to $T$. Combining these two results, we conclude that there exists a small enough increase in $e_g^0$ so that $e_g^t$ increases in all future periods.\end{proof}

\section{Impact of Referral Distribution}
\label{sec:Pchanges}

Recall $P(0|m)$ and $P(2+|m)\equiv\sum_{k\geq2}P(k|m)$  denotes the fraction of workers (of any type) who get no referrals and multiple referrals, respectively, when the average mass of referrals is $m$. $P(0|m)$ and $P(2+|m)$ can be thought of as two measures of referral inequality. 

In this section, we first show how changes in employment levels impact $P(0|m)$ and $P(2+|m)$. We next demonstrate that changes in $P(0|m)$ impact economic productivity. This, together with the preceding analysis of how employment levels impact referrals, is key to proving Proposition~\ref{proposition:group_outcomes}. In Section~\ref{sec:gini}, we note that other metrics such as population-level wage inequality, as captured by the Gini coefficient, are impacted by $P(2+|m)$ as well as $P(0|m)$.  As an increase in one does not imply an increase in the other, these metrics can move in opposing directions. However, as implied by Lemma~\ref{lem:concentration}, $P(0|m)$ and $P(2+|m)$ move together under natural assumptions on the setting.  In such settings, metrics such as immobility, inequality and inefficiency are all mitigated by equalizing employment levels.

\subsection{Impact of Employment Levels on Referral Distribution}
\label{sec:tie}

As suggested in Section~\ref{sec:homophily}, if people belong to groups, one group has an advantage in terms of its initial employment, and people tend to refer people of their own group, then this leads to an aggregate referral distribution with a higher overall fraction of workers who do not get a referral. This is driven by the fraction of workers in a group who do not get a referral, $\hat{P}(0|m)$, being strictly convex in the average number of referrals the group gets, $m$. Here, we show that unequal employment across groups can similarly lead to an aggregate referral distribution with a overall fraction of workers who get multiple referrals. This holds when the fraction of workers in a group who get multiple referrals, $\hat{P}(2+|m)$, is also strictly convex. For a fixed but arbitrary process (e.g., workers making referral uniformly at random), the second derivative of this probability depends on the probability of getting exactly one referral in that group. Thus, for low levels of referrals, strict convexity of the probability of getting multiple referrals for a group follows.\footnote{When referrals are made uniform at random among a group, then $\hat{P}(0|\cdot)$ is strictly convex and $\hat{P}(2+|\cdot)$ is strictly convex on $[0,1]$. When $h_b=h_g=1$ or $n_b,n_g\geq 1$ it follows that any employment levels result in $m_g,m_b\in[0,1]$ as the maximum average number of referrals in each group is at most $1$ so that $P(2+|\cdot)$ is strictly convex on the relevant domain.}

Fix the assumptions and notation discussed in Section~\ref{sec:homophily}. Let $\overline{e}_b,\overline{e}_g$ be the unique employment rates that equate the average number of referrals across groups, $m_b=m_g$.

\begin{lemma}\label{lem:concentration}
Suppose that $\hat{P}(0|\cdot)$ is strictly convex; and $\hat{P}(2+|\cdot)$ is strictly convex for feasible average numbers of referrals $m$. Then:
    \begin{itemize}
        \item the overall fraction of workers who do not get a referral, $P(0|m_b,m_g)$, and the overall fraction of workers who get multiple referrals, $P(2+|m_b,m_g)$, are strictly convex in employment levels; and
        \item both $P(0|m_b,m_g)$ and $P(2+|m_b,m_g)$ are minimized at $\overline{e}_b,\overline{e}_g$.
    \end{itemize}
\end{lemma}

\begin{proof}[Proof of Lemma~\ref{lem:concentration}]
For the first claim, note the respective aggregate probabilities are (nontrivial) linear combinations of $\hat{P}$ so that strict convexity/ concavity is preserved. 
    
For the second claim, take any employment levels $e_b,e_g$. Note that 
$$
    \frac{n_b}{n}\frac{1}{n_b}(h_b{e}_b+(1-h_g){e}_g)+\frac{n_g}{n}\frac{1}{n_g}(h_g{e}_g+(1-h_b){e}_b)=\frac{1}{n}.
$$
Hence, we have
\begin{align*}
    \frac{n_b}{n}\hat{P}(k|\frac{1}{n_b}(h_b{e}_b+(1-h_g){e}_g))+\frac{n_g}{n}\hat{P}(k|\frac{1}{n_g}(h_g{e}_g+(1-h_b){e}_b))\geq \hat{P}(k|1/n),
\end{align*}
for $k \in \{0,2+\}$.
By definition of $\overline{e}_b,\overline{e}_g$, the aggregate probability of not getting a referral given employment levels $\overline{e}_b,\overline{e}_g$ attains this minimum.
\end{proof}

In addition to differences in employment (with similar referral-bias parameters), it could also be that some groups get relatively more referrals per capita due to biases in referrals, $\frac{R_b}{R_g}>\frac{n_b}{n_g}$ (e.g., in some settings males are more strongly biased towards referring males than females are to referring females, as seen in \cite{lalannes2016}). A variant of Lemma~\ref{lem:concentration}---with no employment bias and varying the extent of the referral bias---implies that, under some conditions, a large referral bias (e.g., favoring men) increases the overall fractions of workers who get no or multiple referrals. 

Lastly, the fraction of workers who get no referral and the fraction of workers who get multiple referrals may also be tied together if workers become more likely to refer more-popular workers than a worker at random, since their more-popular worker might be in a better position to return the favor some day. This causes an increase in $P(0|m)$ as well as $P(2+|m)$. 

\subsection{Concentrating Referrals Decreases Productivity}

In this subsection, we presume that firms hire when indifferent ($r=1$). The mixing cases do not change any of the productivity results, but comparing across distributions requires specifying the mixing for each distribution, which has no real consequence but complicates the proofs.\footnote{It can only be consequential when firms do not hire from the pool, as then it can make a difference in the value in the pool;  but the pool value does not affect productivity in that case.} We first investigate how equilibrium productivity changes with the underlying distribution of referrals. Section~\ref{sec:groups} discussed that productivity increases if employment ratios across groups are closer to being population-balanced. Here, we generalize that claim and state it with respect to the overall fraction of workers who do not get a referral.

\begin{proposition}\label{lemma:inequality_and_productivity}
Consider two referral distributions, $P$ and $P'$, and corresponding equilibrium thresholds ${\tilde{v}}$ and ${\tilde{v}}'$, respectively. If $P'$ increases the fraction of workers who do not get a referral compared to $P$ (i.e., $P'(0) \geq  P(0)$), then 
\begin{itemize}
    \item $\E_{\tilde{v}'}[v_i|i\in \pool]\geq \E_{\tilde{v}}[v_i|i\in \pool]$, and so there is
    a weak decrease in the lemons effect and a weak increase in the expected value of workers in the pool;
\item and the total production in the economy associated with $P'$ is less than or equal to that associated with $P$.
\end{itemize}
All comparisons are strict if there is a strict increase in the fraction of workers who do not get a referral (i.e., $P'(0)>P(0)$).
\end{proposition}

\begin{proof}[Proof of Proposition~\ref{lemma:inequality_and_productivity}]
Let $G(\tilde{v},P(0))\coloneqq \E_{\tilde{v}}[v_i|i\in\pool]$; i.e., $$G(\tilde{v},P(0)) \coloneqq   \frac{ P(0) \E[ v_i  ]  +  (1-P(0)) \Pr(v_i<{\tilde{v}})\E [v_i  |  v_i< {\tilde{v}} ] }{P(0) +(1-P(0)) \Pr(v_i<\tilde{v})}.$$
$$= \E[ v_i  ]  \frac{ P(0)   +  (1-P(0)) \Pr(v_i<{\tilde{v}})\frac{\E [v_i  |  v_i< {\tilde{v}} ]}{\E[ v_i  ]} }{P(0) +(1-P(0)) \Pr(v_i<\tilde{v})}.$$

This expression is nondecreasing in $\tilde{v}$.  For any $\tilde{v}$, as $\E_{\tilde{v},r}[v_i|i \in \pool]< \E[v_i]$ by Lemma~\ref{lemma:LE}, it is strictly increasing in $P(0)$ as well.  Therefore
$$
G(\tilde{v},P(0))
\leq G(\tilde{v},P'(0))
\leq G(\tilde{v}',P'(0)),$$
As we've assumed $r=1$, $G(\tilde{v},P(0))=\E_{\tilde{v}}[v_i|i \in \pool]$ and $G(\tilde{v}',P'(0))=\E_{\tilde{v}'}[v_i|i \in \pool]$, this proves the first claimed comparative static.

To prove the second claimed comparative static, note that by the equilibrium condition, employed workers always have higher average productivity than the productivity of unemployed workers in the broader economy (which is assumed to be equal to $\ubar{w}$).  Therefore, to compare the productivities in the two economies, it suffices to compare the mass and productivity of employed workers.  

We consider three cases based on whether firms hire from the pool in equilibrium. 

First suppose that firms hire from the pool in both economies.  Then the mass of employed workers is the same in both economies, and thus the value of the unemployed is the same in both cases.  Note that, by a simple accounting argument, the productivity of employed workers is $n\E[v]$ minus the value of those not hired from the pool $(n-1)\E_{\tilde{v}}[v_i|i\in\pool]$, or $(n-1)\E_{\tilde{v}'}[v_i|i\in\pool]$.  Those not hired from the pool have value given by $G$ above, and so the productivity moves in the opposite direction as $G$ and the result follows from the first claim.   

Next, suppose that they do not hire from the pool in either case.  Then by the equilibrium condition, the hiring threshold in both economies is $\ubar{w}$.  Therefore, the productivity of employed workers in the two economies is the same (namely, $\E[v_i|v_i\geq \ubar{w}]$, which is higher than the minimum wage).  However, as there are weakly more referrals under $P$, there are weakly more employed workers  and so productivity is weakly higher in that economy.  The inequalities hold strictly if $P'(0)>P(0)$ (as then there are strictly more referrals under $P$).

Finally suppose firms hire from the pool in one economy and not the other. Therefore, by the equilibrium condition and the first comparative static, it must be that 
$$
\E_{\tilde{v}}[v_i|i\in\pool]<\ubar{w}\leq \E_{\tilde{v}'}[v_i|i\in\pool].
$$ 
Thus $G(\tilde{v}',P'(0))=\tilde{v}'$ and $G(\tilde{v},P(0))< \tilde{v}=\ubar{w}$. Consider $\overline{P}$ so that $G(\overline{\tilde{v}},\overline{P}(0))=\ubar{w}$.\footnote{Such $\overline{P}$ exists since $G$ is continuous in $P(0)$.} Then $G(\tilde{v}',P'(0))\geq G(\overline{\tilde{v}},\overline{P}(0))\geq G(\tilde{v},P(0))$ and so, by the converse of the first comparative static, it must be that\footnote{The converse follows immediately from the contrapositive and switching $P(0)$ and $P'(0)$.} 
$$
P(0)\leq \overline{P}(0) \leq P'(0).
$$ 
Supposing firms do not hire workers from the pool given $\overline{P}$, we are in the first case and can conclude that productivity is higher given $P$ than $\overline{P}$ (and strictly so if $P(0)<\overline{P}(0)$). Given that firms hire workers from the pool given $\overline{P}$, we are in the second case and can conclude that productivity is weakly higher given $\overline{P}$ than $P'$ (and strictly so if $\overline{P}(0)<P'(0)$). Combining these two inequalities, it follows that productivity is higher in the economy with $P$ than in the economy with $P'$ (and strictly so if $P'(0)>P(0)$).
\end{proof}

Proposition~\ref{lemma:inequality_and_productivity} states that concentrating referrals among a smaller part of the population, leads to a decrease in the lemons effect and lower productivity. Increasing fraction of workers who do not receive a referral reduces the lemons effect since fewer referred workers who are rejected enter the pool (even accounting for the raised hiring threshold, as that raised threshold is itself a reflection of the extent of the lemons effect\footnote{If the lemons effect had actually gone up, then the threshold would have had to fall.}). In fact, the lemons effect moves exactly with the total production in the economy: The average productivity of the workers in the pool equals the average productivity of the unemployed, and thus counter-weights the average productivity of the employed. Intuitively, production worsens since fewer workers are vetted, and less information leads to worse matching.\footnote{This result relies on the assumption that the outside option of workers equals the minimum wage. If $\ubar{w}$ is above outside option value to unemployed workers, then overall production could be higher under $P'$ under some circumstances. That would require that there was no hiring from the pool under $P$, and that the difference between $P'$ and $P$ be large enough to ease the lemons sufficiently to make hiring under the pool attractive under $P'$, and that the gain in value from hiring those workers (the difference between the minimum wage and their outside values) is sufficiently large to offset the loss from reduced vetting. 
More generally, this means that the size of employment could be changing, and here productivity is including outside options and so increased productivity does not necessarily mean increased employment.}

\section{Additional Results on Market Design and Policies}
\label{sec:policyapp}
\subsection{Optimizing Affirmative Action}
\label{sec:optAA}

We have seen that one-time affirmative action policies can have lasting impacts, but how should should a one-time intervention best be implemented? We examine which of the two basic approaches to affirmative action---encouraging more hiring of lower valued greens via referrals compared to hiring more greens from the pool---is more effective.

Specifically, one can increase green employment by either encouraging the hiring of more greens from referrals (which decreases the number of blues hired from the pool), or discouraging the hiring of blues from referrals and thus hiring relatively more greens from the pool. These policies can be enacted by, for instance, paying a small subsidy for each green worker who is hired or else taxing firms that hire a blue worker.\footnote{We presume that a firm can tell if a referred worker is blue or green.} One can also think of a broader class of interventions that include quotas and caps for hires on the referral market.\footnote{Quotas can affect how people are hired in ways that go beyond our model.  For instance, as shown by \cite{chevrot-bianco2021}, quotas on women increased the number of women hired via referrals from family members.  For more discussion on different types of referrals, see \cite{lesterrt2021}.} To capture the trade-off between promoting green workers and demoting blue workers, we compare the effects of increasing the mass of green workers versus decreasing the mass of blue workers hired on the referral market.

Both types of policies distort the optimal hiring decision so that total production decreases in the period where the affirmative action is put in place (recall that the equilibrium in each period is constrained efficient in that it maximizes the production in that period subject to the information structure; see Section~\ref{sec:constrained_efficiency} of the appendix), but have the longer-term impact of increasing productivity in future periods. In general, one of the two policies will dominate the other in the sense that it achieves the same desired goal at a lower reduction in current production.

We illustrate this trade-off in the simple case in which workers are either of high or low value, $v_H$ or $v_L$. Suppose, in the absence of any intervention, there is hiring from the pool, and so $\tilde{v}$ is the expected value in the pool. Also, let $f_g$ be the fraction of workers in the pool that are green. Consider a small intervention---just changing a few hires so that the expected value in the pool is not affected substantially. If one increases the number of green referrals hired, then that puts in a low-value worker in place of a random draw from the pool, so the lost productivity is $\tilde{v}- v_L$ while the gain in green employment is $1-f_g$. If, instead, one decreases the number of blue referrals hired, then that puts in a random draw from the pool in place of a high-value worker, so the lost productivity is $v_H -\tilde{v}$ and the gain in green employment is $f_g$. Which of these is better depends on whether the value of the pool is closer to the high value or the low value as well as the fraction of workers in the pool that are green. 

\begin{proposition}\label{optAA}
The loss in productivity of increasing current green employment by increasing green referral hires is less than that by decreasing blue referral hires if and only if
\[
\frac{1-f_g}{\tilde{v}-v_L}
> \frac{f_g}{v_H-\tilde{v}}.
\]
Furthermore, as long as both policies can achieve the desired employment rates, using one of them as opposed to a combination of them is optimal.
\end{proposition}
The proof of Proposition~\ref{optAA} appears in Section~\ref{online:proofs} of the appendix.

In addition to the type of affirmative action policy, effectively decreasing the number of blue or increasing the number of green referrals hired, one can also optimize over the size of the policy; i.e., how many hiring decisions should be changed. As one increases the size of the affirmative action policy, two things change: 1) the composition of the pool; and 2) the composition of the referral market (both due to not hired blue or hired green referral workers). 

The first is important as those are the workers with which, for example, high-value blue referral workers are replaced with: a firm does not hire their high-value blue referral worker and instead hires from the pool. The second change matters, as those are the ones potentially forced to be hired or not. Interestingly, as in Proposition~\ref{optAA}, changes in the composition of the pool do not always affect the cost-benefit analysis of an affirmative action. Such changes simply end up undoing some of the policy, e.g., hiring a high-value blue referral worker from the pool who was forced to enter the pool. As Proposition~\ref{optAA} only deals with a two-type value distribution, the composition of workers on the referral market does not matter as long as there are, e.g., low-value green workers. 

However, this is not true more generally. When forcing firms to hire green referral workers that otherwise would not have been hired, the optimal policy ensures hiring of such green workers with the highest values. This value is decreasing the more green referral workers are hired so that the productivity decrease of hiring such workers becomes larger.

\subsection{Firing Workers and Inequality}\label{section:firing_workers}

There is a literature on the impacts of a variety of labor market rigidities (e.g., see the discussion and references in \cite{nickell1997}).  Here we present results on the effects of the ease of firing that provide new comparative statics and hypotheses in that direction, by showing how greater ease of firing lowers the lemons effect and reduces the advantage and impact of referrals.  Thus, beyond the usual argument that ease of firing can improve the efficiency of the labor market, we show that it can also lower inequality and improve economic mobility.  

Referrals allow firms to learn about the type of a worker and are thus valuable to firms as well as the workers who are referred. Beyond affirmative action, there are other policies affecting what firms learn about workers that can also reduce inequality. In this section, we consider how the market changes when firms can fire a worker and replace them with another worker from the pool.

In particular, let us suppose that after some time within a single period has elapsed, firms have learned the value of any worker they have hired (they already know the value of a referred worker, so this applies mainly to a worker hired from the pool), and can then choose whether to fire that worker and hire a new one from the pool for the remainder of the period. This makes hiring from the pool more attractive, as there is less expected loss if the worker turns out to have low productivity.

Let $\lambda$ be the fraction of time that is left in the period for which they would get the replacement worker, and $1-\lambda$ be the fraction of the period that has to elapse before a firm can fire a worker. In equilibrium, as we prove below, a firm would never fire a referred worker that they kept initially, so the decision will only be relevant for workers they hired from the initial pool.

The timing is as follows:
\begin{enumerate}
\setlength\itemsep{-0.5em}    
    \item Firms get referrals from their current workers.
    \item Firms can choose whether to (try to) hire a referred worker and make an offer.
    \item Referred workers with offers can choose to accept any of those offers (in which case they exit with the firm) or to go to the pool; all other workers go to the pool.
    \item Firms that did not hire a referred worker can choose whether to hire from the pool.
    \item After $1-\lambda$ of the period has elapsed, firms can choose to fire their current worker.
    \item All fired workers return to the pool, joining all other workers who are not employed.
    \item Firms that have fired a worker can choose to hire from the new pool.
\end{enumerate}

We refer to these pools as pool 1 and pool 2, respectively. We assume that the distribution of values is high enough so that firms prefer to hire from these pools rather than to have no worker at all in order to simplify the exposition; the full details are worked out in Section~\ref{online:proofs}.

A firm's production is given by $1-\lambda$ of the value of its worker before any firing, and then $\lambda$ of the value of its worker after any firing and rehiring decisions are made.\footnote{There is another closely related variant of the model for which the proposition below also holds, which is one in which some fraction $\lambda$ of firms immediately learn the value of their worker (from the pool, and they already know the value from the referral) and can then immediately fire that worker and draw again from the pool. In terms of expected values, this variant of the model looks exactly the same from a firm's perspective, as they just have $\lambda$ weight on an expectation of firing and rehiring, as opposed to a fraction of time. This second variant does have some differences in terms of values in the pools, as fewer firms have an opportunity to fire workers. We will point out when these differences arise; but the basic structure outlined in Lemma~\ref{lemma:lambda_LE} is exactly the same.} An equilibrium is now characterized by a pair of threshold strategies: Referred workers are hired if their value is above $\tilde{v}_1$ and otherwise firms hire from the first pool. Then at the second decision point firms fire workers, when given the opportunity, if the worker's value is below $\tilde{v}_2$. Again, firms can mix with some probability when indifferent.

Let $\E_{\tilde{v}_1,r_1}[v_i |i \in \text{pool }1]$ be defined as in equation~\eqref{tau} and $\E_{\tilde{v}_1,\tilde{v}_2,r_1,r_2}[v_i |i \in \text{pool }2]$ denote the expected value in pool 2 if the hiring threshold on the referral market is given by $\tilde{v}_1$, the firing threshold by $\tilde{v}_2$, firms hire referral workers at the margin with mixing parameter $r_1$ and firms fire workers at the margin with mixing parameter $1-r_2$. The threshold strategies then imply the following two equilibrium conditions:
\begin{align}
\tilde{v}_1&=(1-\lambda)\E_{\tilde{v}_1,r_1}[v_i |i \in \text{pool }1] + \lambda \E_{\tilde{v}_1,r_1}[\max \{v_i,\tilde{v}_2\} |i \in \text{pool }1]; \quad \text{and}\label{eq:cond1} \\
\tilde{v}_2&=\E_{\tilde{v}_1,\tilde{v}_2,r_1,r_2}[v_i |i \in \text{pool }2] \label{eq:cond2}.
\end{align}

Let $\tilde{v}$ denote the hiring threshold for the model without firing (so that $\lambda=0$). For $\lambda \in (0,1]$, not hiring a referred worker and instead going to pool 1, now features an option value: The firm randomly draws a worker but may get a second draw to replace the worker if the worker's value is too low; a draw from pool 2. This leads to firms adopting a higher threshold on the referral market compared to the base model. The higher initial hiring threshold, in turn, lessens the lemons effect in pool 1 and so the pool 1 value is actually higher than $\tilde{v}$. Pool 2, however, has a worse lemons effect since it then includes all fired workers as well as rejected referral workers. The fact that the second threshold is less than $\tilde{v}$ takes some proof, but is true. 

\begin{lemma}\label{lemma:lambda_LE}
There are unique thresholds $(\tilde{v}_1,\tilde{v}_2)$ and mixing parameters $(r_1,r_2)$ with $r_2$ arbitrary satisfying equations \eqref{eq:cond1} and \eqref{eq:cond2}.
For $\lambda>0$, they satisfy
$
	\tilde{v}_2 < \tilde{v} < \tilde{v}_1.
$
\end{lemma}
\noindent Note that our base model is nested in this formulation with $\lambda = 0$. The proof of Lemma~\ref{lemma:lambda_LE} appears in Section~\ref{online:proofs} of the appendix.

The equilibrium thresholds, as before, are unique, so that we can perform comparative statics with respect to unique equilibrium quantities. Firms have a higher hiring threshold on the referral market when $\lambda>0$ and rely more on hiring from the pool. The additional weight on hiring from the pool increases the opportunities for disadvantaged groups---e.g., the green workers---and improves their employment prospects. The willingness of firms to engage in further search decreases the production in the first part of the period.

However, as the proposition below shows, overall production---the weighted production in the two parts of the period---increases with $\lambda$. Thus, efficiency and equality are always greater when there are opportunities for firms to fire their worker.

\begin{proposition}\label{proposition:firing_workers_lambda}
For any $\lambda>0$, the employment rate bias is less than, and total production is greater than, what it would have been without the opportunity to fire ($\lambda=0$).  In terms of timing within the period: (i) both the employment bias in favor of blue workers and productivity measured before the firing stage are decreasing in $\lambda$; and (ii) the productivity at the end of the period is increasing in $\lambda$ while the employment bias in favor of blue workers at the end of the period changes ambiguously. 
\end{proposition}
The proof of Proposition~\ref{proposition:firing_workers_lambda} appears in Section~\ref{online:proofs} of the appendix.

Comparing productivity and equality between two different positive values of $\lambda$ is nuanced. While efficiency unambiguously increases the greater the opportunities for firms to fire a worker, the employment rate bias at the end of the period may behave nonmonotonically in $\lambda$. More use of pool 1 can lead to a worse lemons effect in pool 2 as firms engage in more search. This can disadvantage the greens. As a result, the overall implication for the employment rate bias measured at the end of the period is ambiguous.

\subsection{The Impact of Macroeconomic Conditions on Productivity and Inequality}
\label{sec:macro}

Before closing, let us mention one more comparative static which illustrates how the effects of the model change with outside conditions. We examine how macroeconomic conditions, such as the number of jobs available, affect productivity and inequality among those being hired in any given period.\footnote{For a comparative static in how the size of a network affects hiring via referrals (in a different model), see \cite{calvo2005job}.} To answer this, we hold the network of connections fixed, but randomly remove some firms from the job market (so that a connection that resulted in a referral may now be useless). In particular, if a mass $\kappa<1$ of firms are no longer hiring, what is the impact on the job market? Of course, this decreases overall employment.\footnote{This presumes that workers were being hired from the pool initially.  If not, there  can be very particular situations in which the reduced lemons effect makes it worthwhile to hire from the pool, which could reverse this effect.}  The more subtle question is what happens to the average productivity per employed worker and to the inequality in the society.  

First, let us consider the productivity per employed worker. The drop in firms decreases the fraction of workers who get at least one referral, i.e., who are looked at on the referral market, $1-P(0)$, by some positive amount (at most $\kappa/n$). As a result, the lemons effect decreases and the remaining firms now have better outside options and their production, so productivity per worker, increases. 

Second, firms that do remain are less likely to face competition from other firms for their referred workers.  This leads to a decrease in average wages because of the reduced competition and also the lower lemons effect and improved alternative of hiring from the pool. 

These results are summarized in the following result, in which we presume that firms would hire from the pool if $\kappa=0$.  It is proven as Proposition~\ref{prop:macro} in Section~\ref{online:proofs} of the appendix.

Suppose a mass $\kappa<1$ of firms no longer hire on the job market. Then:
\begin{itemize}
\setlength\itemsep{-0.5em}    
    \item total production decreases;
    \item production per employed worker increases; and
    \item the old wage distribution of wages first-order stochastically dominates the new one.
\end{itemize}
All comparisons are strict if $\kappa>0$.

Are different groups differentially affected by such downturns? There are two effects.  One is the reduction in the lemons effect: this unambiguously reduces inequality across groups as it disproportionately lowers the advantage of referrals both in terms of wages and the chances of being employed. The second effect is that workers lose referrals, which actually relatively hurts the disadvantaged group more, since more of them who are getting referrals are only getting one referral and lose that crucial referral. In particular, if $\hat{P}(0|m)$ is convex in $m$ (as is naturally the case for most distributions, as previously discussed), then for purely homophilous groups, the disadvantaged group, rather than the advantaged group, loses a larger fraction of workers looked at on the referral market.

\subsection{The Impact of Concentrating Referrals on Wage Inequality}\label{sec:gini}

In Section~\ref{sec:ineq_and_immobility}, we studied how employment bias from unequal referrals and homophily leads to inequality across groups (Proposition~\ref{proposition:group_outcomes}). However, this analysis does not speak to aggregate measures of inequality in society (ignoring group identities) as a whole and how they respond to changes in the referral distribution. This is the subject of this section: the impact on wage inequality (which in this model is equivalent to income inequality; more comments on this below) from changes in the referral distribution.

The first thing to note is the key factor in determining the wage distribution in our model is the fraction of workers who get more than one referral: Those are the workers who have some competition for their services and earn above the minimum wage. The fraction of workers who get no referral still determines the expected value from hiring from the pool, and thus impacts wages that people with multiple referrals obtain, so it is also involved but with a different (marginal) effect.

To measure inequality, we use the Gini coefficient of wages (a formal definition for our setting follows shortly in equation~\eqref{giniexp}). As comparing distributions is generally an incomplete exercise, using such a coefficient allows one to make comparisons of partially ordered distributions. However, even using this standard one-dimensional measure, and looking at simple settings, we will see that inequality can move in different ways from the same comparative statics, depending on the specifics of the environment.

One might expect that increasing the fraction of workers who get multiple referrals would be sufficient to increase inequality, but things are not so direct. First, even if few workers are high-wage earners, increasing the size of that group can actually increase or decrease inequality, depending on relative wages and the relative size of the group to begin with, as the Gini coefficient makes relative comparisons. In addition, increasing the fraction of workers who get multiple referrals can alleviate the lemons effect if accompanied by a decrease in the fraction of workers with no referral. This tends to decrease the wages of the workers who have multiple referrals, lowering the high wage, which can then decrease inequality. Thus, there can be different sorts of countervailing effects.

To get a characterization of when inequality is raised by increasing the fraction of workers who get multiple referrals, we consider an environment in which workers are either of high or low value, $v_H$ or $v_L$, with fraction $f_H$ of the population having the high value. As the minimum wage equals the value of workers' outside options, this implies that there are just two levels of wages. All of the countervailing effects described above already exist with just two types, and this simplification makes it easier to see the intuition.

Let $w_H \equiv v_H-\tilde{v}+\ubar{w}$ denote the equilibrium wage of a worker with high value (the ``high wage''), let $\pi_H \equiv P(2+)f_H$ denote the fraction of workers who earn the high wage,  and let $\pi_L \equiv (1- P(2+)f_H)$ be the remaining fraction of workers, who all earn the minimum wage, $\ubar{w}$, (the ``low wage''). Finally, let  $W_L \equiv \frac{\ubar{w}}{w_H}$ denote the low wage relative to the high wage. One can write the Gini coefficient, denoted by $Gini$, of an economy as follows:\footnote{The expression is based on the definition of the Gini coefficient as the mean absolute difference in wages, normalized by twice the average wage. This definition is equivalent to the more standard one in terms of the area underneath the Lorenz Curve \citep{SenEconomicInequality1973}.}
\begin{equation}
\label{giniexp}
Gini = \frac{\pi_H\pi_L(w_H-\ubar{w})}{\pi_H w_H+\pi_L \ubar{w}} =
\frac{ \pi_H \pi_L (1-W_L)}{ \pi_H+ \pi_L W_L }= \frac{ \pi_H (1-\pi_H) (1-W_L)}{ \pi_H+ (1-\pi_H) W_L }.
\end{equation}

Let us consider what happens when the fraction of workers who receive the high wage increases ($\pi_H$ increases), which is a consequence of increasing the fraction of workers who get get multiple referrals. Straightforward calculations show that:
\begin{equation*}
\frac{\partial Gini}{\partial \pi_H} = \frac{(1-W_L)(W_L\pi_L^2-\pi_H^2)-\frac{\partial W_L}{\partial \pi_H}\pi_H\pi_L}{(\pi_H+\pi_LW_L)^2}.
\end{equation*}
The associated change in the Gini coefficient consists of two parts: the effect of changing the fraction of workers who receive the high wage, and then the effect of changing the wage via the lemons effect.

The first part  (ignoring the term including $\frac{\partial W_L}{\partial \pi_H} $) is positive for low $\pi_H$ and high $W_L$, but then becomes negative as $\pi_H$ increases and $W_L$ decreases. However, then the overall expression is impacted by the  $\frac{\partial W_L}{\partial \pi_H}$ term, which accounts for the lemons effect. The change in the lemons effect depends on both the probability of not getting a referral as well as multiple referrals.

To fully sign the $\frac{\partial W_L}{\partial \pi_H}$ term, we also need to know what happens to the fraction of workers who do not get a referral. In particular, the change in $W_L$ comes from the change in $w_H$, which is exactly governed by the fraction of workers who do not get a referral. Knowing how $P(2+)$ changes, tells us about $\pi_H$, but we need to know how $P(0)$ changes to determine the change in $w_H$. If the change in the referral distribution is comprised of both an increase in the fraction of workers who do not get a referral and the fraction of workers who get multiple ones, then the increase in $\pi_H$ is accompanied by an increase in $P(0)$ which, using Proposition~\ref{lemma:inequality_and_productivity}, decreases the lemons effect and thus the high wage so that $\frac{\partial W_L}{\partial \pi_H}>0$. This then counteracts the impact of the first term for low $\pi_H$ and can reverse the change in the Gini coefficient.

These results are summarized in the following proposition, which should be clear from the above discussion, so we omit a formal proof. Let $I_H=\pi_H w_H$ and $I_L=\pi_L w_L$ denote the total wage going to the high- and low-wage workers, respectively.

\begin{proposition}
\label{GiniChange}
The Gini coefficient decreases with an increase in the fraction of workers who do not get a referral (an increase in $P(0)$), holding $P(2+)$ constant. And, holding $P(0)$ constant, the Gini coefficient increases with an increase in the fraction of workers who get multiple referrals (an increase in $P(2+)$) if and only if $I_L\pi_L> I_H\pi_H$.
\end{proposition}

The comparison $I_L\pi_L> I_H\pi_H$ captures the relative total wages, population weighted, which is important in determining whether the inequality is due to most of the society's income coming from the low wage, which holds for low values of $\pi_H$ and ${w_H}$, or the reverse, which determines whether things are getting more or less equal.

Proposition~\ref{GiniChange} implies that, e.g., a mean-preserving spread in $P$, such that the probability of not getting a referral gos down and the probability of getting multiple referrals goes up  (presuming $I_L\pi_L> I_H\pi_H$),  pushes up the fraction of high wage earners and also increases their wage due to an increased lemons effect, thus increasing the Gini coefficient. If instead (again presuming that  $I_L\pi_L> I_H\pi_H$), we consider a first-order dominance shift in $P$, then there are countervailing forces: more people earn the higher wage, but that wage goes down due to an improved value of the pool. Either force can dominate depending on the particular parameters.

In Section~\ref{online-appendix:results} of the appendix, we discuss what happens when firms' profits are also included in the calculations of inequality.

\subsection{Proofs of Results in Appendix~\ref{sec:policyapp}}\label{online:proofs}

\begin{proof}[Proof of Proposition~\ref{optAA}]
A combination of a quota and cap policy (referred to as the policy) can be represented by the sets $A_b$ and $A_g$ where $A_b$ is the set of blue workers who are not hired on the referral market because of the policy and $A_g$ is the set of green workers who are hired on the referral market because of the policy. As low-value blue workers would also not be hired on the referral market, $A_b$ consists of high-value blue workers. Similarly, $A_g$ consists of low-value green workers.

We want to understand the differences in total production and employment by group with and without the policy. To this end, we partition the set of workers into several subsets whose aggregate composition of workers stays constant with or without the policy. We then track the extend to which each of these sets of workers is employed. To be clear, the realization of randomization may imply that individual workers belong to different sets with or without the policy. However, workers can always be grouped into such subsets so that their composition in terms of value and group stays constant.

The pool under the policy differs from the pool without it in two ways. First, workers in $A_b$ enter the pool. Second, workers in $A_g$ do not enter pool. Letting $B$ denote the set of workers who enter the pool with and without the policy, the pool is given by the union of $A_g$ and $B$ without the policy and by the union of $A_b$ and $B$ with the policy. We represent these sets graphically in Figure \ref{fig1} where $A_g$ is represented by the dotted rectangle, $A_g$ by the dashed rectangle and $B$ by the solid rectangle.

Due to uniform sampling from the pool, it must be that the same fraction of workers in $A_g$ and $B$ is hired without the policy and of $A_b$ and $B$ with the policy. We graphically represent this fact by placing the aforementioned rectangles side-by-side and horizontally slicing through them (the dashed lines) with the top area representing the set of unemployed workers and the bottom area the set of workers hired from the pool. 

Note that the size of the set of unemployed needs to be constant, it equals $n-1$, regardless of whether the policy is in place or not. Thus, in Figure \ref{fig1}, the size of the dotted area equals size of the gray area in each subfigure respectively. 

Suppose first that more low-value green workers are hired on the referral market than high-value blue workers are not, i.e., $|A_g|\geq |A_b|$. To compare total production and employment by group with and without the policy, we do the following calculation. Due to the policy, workers in the gray region in Figure \ref{fig1a} are displaced by workers in the dotted region. We separately analyze the impact of displacing workers in the gray region inside $B$ and workers in the gray region outside $B$.  For the first calculation, let $\tilde{v}^*$ denote the average value of a worker in $B$ and $f^*_g$ the fraction of workers in $B$ that are green. Replacing a mass of such workers with low-value green workers increases the employment of green workers (as $f^*_g\leq 1$) and decreases total production (as $\tilde{v}\geq v_L$). The rate at which this occurs, the gain in green employment divided by the loss in production is given by 
\begin{equation}\label{eq:1a}
    \frac{1-f_g^*}{\tilde{v}^*-v_L}.
\end{equation}
We can express $\tilde{v}^*$ and $f_g^*$ in terms of $\tilde{v}$ and $f_g$, the average productivity and the fraction of green workers in the pool without the policy in place, i.e., in $B\cup A_g$
\begin{align*}
	\tilde{v} = \frac{|B|\tilde{v}^*+|A_g|v_L}{|B|+|A_g|} &\implies \tilde{v}^* = \frac{(|B|+|A_g|)\tilde{v}-|A_g|v_L}{|B|}; \\
	f_g = \frac{|B|f_g^*+|A_g|\cdot 1}{|B|+|A_g|} &\implies f_g^* = \frac{(|B|+|A_g|)f_g-|A_g|\cdot 1}{|B|}.
\end{align*}
\eqref{eq:1a} then simplifies to 
\begin{equation}\label{eq:1b}
    \frac{1-f_g}{\tilde{v}-v_L}.
\end{equation}
For the second calculation, note the rest of the workers in the dotted area are replacing workers in the intersection of the gray area with $A_b$, the set of high-value blue workers that are not hired because of the policy. The gain in green employment divided by the loss in production for this change in the employment composition is simply
\begin{equation}\label{eq:1c}
    \frac{1}{v_H-v_L}.
\end{equation}
Thus, it must be that the total gain in employment divided by the total loss in production due to the policy is sandwiched by \eqref{eq:1b} and \eqref{eq:1c}.

Let us now suppose that more low-value green workers are hired on the referral market than high-value blue workers are not, i.e., $|A_g|\leq |A_b|$. Now, due to the policy, workers in the gray region in Figure \ref{fig1b} are displaced by workers in the dotted region. To describe the dotted region in terms of primitives ($v_L)$  and characteristics of the pool without the policy ($\tilde{v}$ and $f_g$), we divide it up into two rectangles as shown in Figure \ref{fig1c}. Recall that the union of $A_g$ and $B$ is simply the pool in the absence of the policy. Thus, the average value of a worker in the bottom rectangle in Figure \ref{fig1c} and the fraction of such workers that are green are given by $\tilde{v}$ and $f_g$ respectively. The increase in employment of green workers (as $f_g\geq 0$) divided by the loss in production (as $v_H\geq \tilde{v})$ when replacing workers in $A_b$ with workers in the aforementioned rectangle is thus given by 
\begin{equation}\label{eq:1d}
    \frac{f_g}{v_H-\tilde{v}}.
\end{equation}
The rest of the workers in gray area are replacing workers in the top rectangle in Figure \ref{fig1c}, low-value green workers. As before, the increase in green employment over the loss in total production is given by \eqref{eq:1c}. Thus, it must be that the total gain in employment divided by the total loss in production due to the policy is sandwiched by \eqref{eq:1c} and \eqref{eq:1d}.

It is clear from the above calculations that the increase in green employment over the loss in total production for $|A_g|>|A_b|=0$, $|A_g|=0<|A_b|$ and $|A_g|=|A_b|$ is given by $\frac{1-f_g}{\tilde{v}-v_L}$, $\frac{f_g}{v_H-\tilde{v}}$ and $\frac{1}{v_H-v_L}$ respectively. Furthermore, $\frac{1}{v_H-v_L}$ is sandwiched by $\frac{1-f_g}{\tilde{v}-v_L}$ and $\frac{f_g}{v_H-\tilde{v}}$ so that it suffices to consider a pure quota or cap policy only.

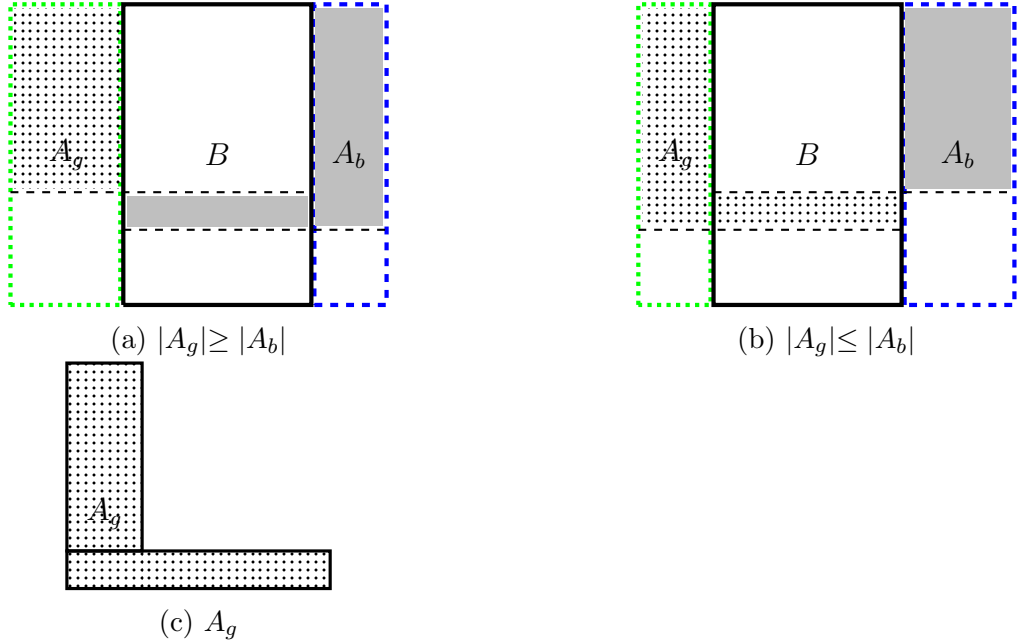
\begin{figure}[ht!]
\begin{subfigure}{.5\textwidth}
  \centering
\begin{tikzpicture}[scale=.5]
\def \shift {0.08}
\draw[ultra thick,dashed] (8+\shift,0) -- (10,0) -- (10,8) -- (8+\shift,8) -- (8+\shift,0);
\draw[ultra thick,dotted] (0,0) -- (3-\shift,0) -- (3-\shift,8) -- (0,8) -- (0,0);
\draw[ultra thick] (3,0) -- (8,0) -- (8,8) -- (3,8) -- (3,0);

\def \shift {0.1}
\fill[pattern = dots] (0+\shift,3+\shift) rectangle (3-\shift,8-\shift);
\fill[lightgray] (3+\shift,3-\shift) rectangle (8-\shift,2+\shift);
\fill[lightgray] (8+\shift,2+\shift) rectangle (10-\shift,8-\shift);

\draw (1.5,4) node {$A_g$};
\draw (5.5,4) node {$B$};
\draw (9,4) node {$A_b$};
\draw[thick,dashed] (0,3) -- (8,3);
\draw[thick,dashed] (3,2) -- (10,2);
\end{tikzpicture}
\caption{$|A_g|\geq |A_b|$}
\label{fig1a}
\end{subfigure}
\begin{subfigure}{.5\textwidth}
  \centering
\begin{tikzpicture}[scale=.5]
\def \shift {0.08}
\draw[ultra thick,dashed] (7+\shift,0) rectangle (10,8);
\draw[ultra thick,dotted] (0,0) rectangle (2-\shift,8);
\draw[ultra thick] (2,0) rectangle (7,8);

\def \shift {0.1}
\fill[pattern = dots] (0+\shift,2+\shift) rectangle (2-\shift,8-\shift);
\fill[pattern = dots] (2+\shift,3-\shift) rectangle (7-\shift,2+\shift);
\fill[lightgray] (7+\shift,3+\shift) rectangle (10-\shift,8-\shift);

\draw (1,4) node {$A_g$};
\draw (4.5,4) node {$B$};
\draw (8.5,4) node {$A_b$};
\draw[thick,dashed] (0,2) -- (7,2);
\draw[thick,dashed] (2,3) -- (10,3);
\end{tikzpicture}
\caption{$|A_g|\leq |A_b|$}
\label{fig1b}
\end{subfigure}
\begin{subfigure}
  {.5\textwidth}
  \centering
\begin{tikzpicture}[scale=.5]
\def \shift {0}
\fill[pattern = dots] (0+\shift,2+\shift) rectangle (2-\shift,8-\shift);
\fill[pattern = dots] (2+\shift,3-\shift) rectangle (7-\shift,2+\shift);
\draw (1,4) node {$A_g$};
\draw[very thick] (0,3) rectangle (2,8);
\draw[very thick] (0,2) rectangle (7,3);
\end{tikzpicture}
\caption{$A_g$}
\label{fig1c}
\end{subfigure}
\caption{A graphical representation of the pool with and without the policy.}
\label{fig1}
\end{figure}
\end{proof}

We make use of the following lemma in some of the remaining proofs.

\begin{lemma}\label{lemma:monotonicity}
Let $\tilde{v}$ denote the unique equilibrium threshold and $v$ a generic hiring threshold. Then $\min_r\E_{v,r}[v_i|i \in \pool]\geq \max_r\E_{v',r}[v_i|i \in \pool]$ if $v<v'<\tilde{v}$ and $\max_r\E_{v,r}[v_i|i \in \pool]\geq \min_r\E_{v',r}[v_i|i \in \pool]$ if $\tilde{v}<v<v'$. In particular, if $F$ is continuous, then \eqref{tau} is decreasing in $v$ for $v<\tilde{v}$ and increasing otherwise.

As total production is inversely related to the average value in the pool, it is first increasing and then decreasing in the hiring threshold with the same caveat at mass points of the distribution.
\end{lemma}

\begin{proof}[Proof of Lemma~\ref{lemma:monotonicity}]
As an immediate corollary of Lemma~\ref{eqw}, for all threshold $v$, 
\begin{equation*}
    \E_{v,r}[v_i|i \in \pool]>v \iff v<\tilde{v}.
\end{equation*}
Thus, if $v<v'<\tilde{v}$, then the pool with threshold $v'$ consists of the union of marginal workers added to the pool through increases in the threshold each of which with values lower than the expected value in the pool when they were added. Thus, the result follows with \eqref{tau} increasing for values larger than $\tilde{v}$ shown analogously.\end{proof}

\begin{proof}[Proof of Lemma~\ref{lemma:lambda_LE}]
Existence of equilibrium thresholds is proven analogously as in Lemma~\ref{eqw}. Similarly, equations \eqref{eq:cond1} and \eqref{eq:cond2} are simply optimality conditions that need to be satisfied for firms to adhere to the equilibrium thresholds thus showing equilibrium.

Next, to order the thresholds, suppose that $\tilde{v}_2\geq \tilde{v}_1$. Pool 2 consists of workers from pool 1 with some workers with values above $\tilde{v}_2$ removed and (referral) workers with values between $\tilde{v}_1$ and $\tilde{v}_2$ added. The first change implies a decrease in the average value in the pool, whereas the second an increase. However, adding a mass of workers with values between $\tilde{v}_1$ and $\tilde{v}_2$ cannot result in the average value in the pool to increase to $\tilde{v}_2$. The nondegeneracy of $F$ then implies that $\tilde{v}_1$ must in fact be strictly greater than $\tilde{v}_2$.

Suppose that $\tilde{v}_1\leq \tilde{v}$. As when the referral threshold equals $\tilde{v}$, the expected value in pool 1 equals $\tilde{v}$ and is otherwise greater by Lemma~\ref{lemma:monotonicity}, it is clear that \eqref{eq:cond1} cannot be satisfied; again, with strictness coming from the nondegeneracy of $F$.

Suppose that $\tilde{v}_2\geq \tilde{v}$. In equilibrium, total production is given by
\begin{equation}\label{eq:prodfir}
	n(1-P(0))\E[\max\{v_i,\tilde{v}_1\}] + (1-n(1-P(0)))\tilde{v}_1.
\end{equation}
Similarly, total production for $\lambda =0$ is given by
\begin{equation*}
	n(1-P(0))\E[\max\{v_i,\tilde{v}\}] + (1-n(1-P(0)))\tilde{v},
\end{equation*}
and thus clearly total production is greater in the former case $\tilde{v}_1>\tilde{v}$.

However, production before firing must be lower as a referral threshold of $\tilde{v}$ maximizes (immediate) production by Lemma~\ref{lemma:monotonicity}. Furthermore, as $\tilde{v}_2\geq \tilde{v}$, production after firing is also lower when firing is permitted by a simple accounting exercise as in Lemma~\ref{lemma:inequality_and_productivity}; a contradiction. Thus, $\tilde{v}_2< \tilde{v} <\tilde{v}_1$.

To show uniqueness for $\tilde{v}_1$, suppose the contrary and consider two hiring thresholds $\tilde{v}_1$ and $\tilde{v}_1'$ and suppose without loss that $\tilde{v}<\tilde{v}_1<\tilde{v}_1'$. Total production must be higher under $\tilde{v}_1'$ than under $\tilde{v}_1$ as is evident from \eqref{eq:prodfir}. By Lemma~\ref{lemma:monotonicity}, production before firing is lower under $\tilde{v}_1'$ than under $\tilde{v}_1$ as $\tilde{v}<\tilde{v}_1<\tilde{v}_1'$. Hence, it must be that $\tilde{v}_2'<\tilde{v}_2$ for total production to be the same under the two thresholds. But then \eqref{eq:cond1} cannot be satisfied for both $\tilde{v}_1$ and $\tilde{v}_2$ and $\tilde{v}_1'$ and $\tilde{v}_2'$ as the both terms on the right-hand side are larger given $\tilde{v}_1$ and $\tilde{v}_2$ than given $\tilde{v}_1'$ and $\tilde{v}_2'$ whereas $\tilde{v}_1<\tilde{v}_1'$ by assumption.

Lastly, let us show that $\tilde{v}_2$ is unique. By the uniqueness of $\tilde{v}_1$, it must be that there are $r_1,r_1'$ with $r_1> r_1'$ (without loss) so that $\tilde{v}_1$ with $r_1$ and some $\tilde{v}_2$ as well as $\tilde{v}_1$ with $r_1'$ and some $\tilde{v}_2'$ with $\tilde{v}_2\neq \tilde{v}_2'$ satisfy \eqref{eq:cond1} and \eqref{eq:cond2}. As $\tilde{v}_1>\tilde{v}$, it must be that production before the firing is higher given $\tilde{v}_1$ and $r_1$ than $\tilde{v}_1$ and $r_1'$ and 
\begin{equation*}
    \E_{\tilde{v}_1,\tilde{v}_2,r_1,r_2}[v_i| i\in \text{pool }1]<\E_{\tilde{v}_1,\tilde{v}_2',r_1',r_2'}[v_i| i\in \text{pool }1]
\end{equation*}
by an argument similar to that of Lemma~\ref{lemma:monotonicity}. As total production in equilibrium is determined by $\tilde{v}_1$ according to \eqref{eq:prodfir}, it must be the same for both sets of thresholds. Hence, it must be that $\tilde{v}_2<\tilde{v}_2'$, so that the comparison of production reverses after firing. However, then equilibrium condition \eqref{eq:cond1} cannot be satisfied for both sets of thresholds as both terms on the right-hand side of \eqref{eq:cond1} are larger for $\tilde{v}_1$ with $r_1'$ and $\tilde{v}_2'$ than with $r_1$ and $\tilde{v}_2$ whereas the left-hand side, $\tilde{v}_1$, is constant.
\end{proof}

\begin{proof}[Proof of Proposition~\ref{proposition:firing_workers_lambda}]
First, we show that $\tilde{v}_1$ is strictly increasing in $\lambda$. Fixing $\tilde{v}_1$, and thus $\tilde{v}_2$ , the right-hand side of \eqref{eq:cond1} strictly increases in $\lambda$ (as $\tilde{v}_2$ is bounded away from the lowest possible value as $P(0)>0$). Thus, it must be that $\tilde{v}_1$ strictly increases with $\lambda$ as $\tilde{v}_1$ is unique by Lemma~\ref{lemma:lambda_LE}. In equilibrium, total production is given by \eqref{eq:prodfir} and in particular strictly increasing in $\tilde{v}_1$, so that is is also strictly increasing in $\lambda$. 

Next, we show that $\tilde{v}_2$ is strictly decreasing in $\lambda$. Let $\lambda'>\lambda$ and denote the equilibrium thresholds by $(\tilde{v}_1',\tilde{v}_2')$ and $(\tilde{v}_1,\tilde{v}_2)$ respectively. We know that $\tilde{v}_1'> \tilde{v}_1$. To reach a contradiction suppose that $\tilde{v}_2'\geq \tilde{v}_2$.

We study how production after firing differs given $\lambda$ and $\lambda'$ with $\lambda,\lambda'>0$. To this end, we partition the set of workers into several subsets whose aggregate composition of workers stays constant given $\lambda$ and $\lambda'$. We then track the extend to which each of these sets of workers is employed. To be clear, the realization of randomization may imply that individual workers belong to different sets depending on whether we consider $\lambda$ or $\lambda$'. However, workers can always be grouped into such subsets so that their composition in terms of values stays constant.

Under $\lambda'$, as $\tilde{v}_1'>\tilde{v}_1$, an additional set of workers enter the pool. In Figure \ref{fig:2}, the thick solid rectangle represents workers who enter the pool under both $\lambda$ and $\lambda'$ whereas the thick dashed rectangle represents workers who only enter the pool under $\lambda'$, i.e., those with values between $\tilde{v}_1$ and $\tilde{v}_1'$. Due to uniform sampling from pool 1, it must be that the same fraction of workers in both rectangles is hired from the aforementioned sets. We graphically represent this fact by placing the rectangles side-by-side and, given $\lambda'$, horizontally slice through them (the top dashed lines) with the top area representing the set of unemployed workers and the bottom area the set of workers hired from pool 1. The second dashed line (the bottom dashed line) represents the random hiring from pool 1 given $\lambda$. Note that the ordering of the dashed lines is deliberate and follows from the fact that an equal increase in firms hiring from pool 1 and workers entering pool 1 implies a higher fraction of workers from pool 1 hired. Furthermore, the size of the set of unemployed before firing needs to be constant, it equals $n-1$, given $\lambda$ or $\lambda'$.  Thus, in Figure \ref{fig:2}, the size of area $C$ equals that of area $D$ (all labels in Figure \ref{fig:2} correspond to the smallest rectangle containing them). A final set of workers consists of those who are hired on the referral market under both $\lambda$ and $\lambda'$. Their contribution to total production after firing is obviously constant under $\lambda$ and $\lambda'$.

Consider the firms hiring workers in $A$ from pool 1. The distribution of values in set $A$ is constant by construction; and, by assumption, a firm at first hiring such worker can replace the worker with a worker with value $\tilde{v}_2'\geq \tilde{v}_2$ under $\lambda'$ and with a worker with value $\tilde{v}_2$ under $\lambda$. Thus, average production after firing of firms hiring workers in $A$ is larger given $\lambda'$.

Consider the firms hiring workers in $B$ from pool 1. All such workers would have been hired on the referral market under $\lambda$ and are drawn from pool 1 under $\lambda'$. Note that no such worker is fired (see Lemma~\ref{lemma:lambda_LE}). Thus, their contribution to total production after firing is constant under $\lambda$ and $\lambda'$. 

Finally, let us compare how firms fare that under $\lambda$ hire workers in $C$ from pool 1 and under $\lambda'$ workers in $D$ from pool 1. We have 
\begin{equation}\label{eq:firing1}
    \E[v_i|i \in D]\leq \tilde{v}_1' < \E_{\tilde{v}_1',r_1'}[\max\{v_i,\tilde{v}_2'\}|i \in \text{pool 1}],
\end{equation}
where the first inequality follows as workers in $D$ have values between $\tilde{v}_1$ and $\tilde{v}_1'$ and the second from \eqref{eq:cond1}. Note the last expectation in the equation above is taken over workers in pool 1 given $\lambda'$. As pool 1 given $\lambda'$ is the union of pool 1 given $\lambda$ and sets $B$ and $D$, the following equation must hold
\begin{multline}\label{eq:firing2}
    \E_{\tilde{v}_1',r_1'}[\max\{v_i,\tilde{v}_2'\}|i \in \text{pool 1}]= \frac{|B|+|D|}{|A|+|B|+|C|+|D|+|E|}\E[v_i|i\in B \cup D] \\ +  \frac{|A|+|C|+|E|}{|A|+|B|+|C|+|D|+|E|}\E_{\tilde{v}_1,r_1}[v_i|i \in \text{pool 1}],
\end{multline}
where the first expectation is taken over workers in pool 1 given $\lambda'$ and the last over workers in pool 1 given $\lambda$. As $\E[v_i|i \in D]=\E[v_i|i \in B \cup D]$, \eqref{eq:firing1} and \eqref{eq:firing2} imply
\begin{equation*}
    \E[v_i| i \in D]<\E_{\tilde{v}_1,r_1}[\max\{v_i,\tilde{v}_2\}|i \in \text{pool 1}],
\end{equation*}
once more, with the expectation taken over workers in pool 1 given $\lambda$. The left-hand side is the average production after firing of the firms hiring workers in $D$ under $\lambda$ while the right-hand side is the average production after firing of firms hiring workers in $C$ under $\lambda'$. Thus, the production after firing the  of the latter equally sized mass of firms is larger.

We have then considered the set of all firms and production after firing is strictly greater under $\lambda'$ than under $\lambda$. As production is inversely related to the average value of the unemployed workers, it must be that $\tilde{v}_2'<\tilde{v}_2$; contrary to our assumption.
Thus, it must be $\tilde{v}_2'<\tilde{v}_2$.

\begin{figure}
    \centering
    \begin{tikzpicture}[scale=.5]
    \draw[ultra thick] (0,0) rectangle (6,8);
    \draw[ultra thick, dashed] (6,0) -- (10,0) --  (10,8) -- (6,8);
    \draw[thick, dashed] (0,5) -- (6,5);
    \draw[thick, dashed] (0,6) -- (10,6);
    \draw (3,2.5) node {$A$};
    \draw (8,3) node {$B$};
    \draw (3,5.5) node {$C$};
    \draw (8,7) node {$D$};
    \draw (3,7) node {$E$};
    \end{tikzpicture}
    \caption{Pool 1}
    \label{fig:2}
\end{figure}

The employment rate of green workers before the firing stage is given by 
\begin{equation}\label{eq:emplgreen}
    \frac{R_g}{n_g}+\left(1-\frac{R_g}{n_g}\right)\frac{1-(R_b+R_g)}{n-(R_b+R_g)},
\end{equation}
where $R_b$ and $R_g$ are the masses of blue and green workers that are hired on the referral market, respectively. As blue workers have a lower probability of not getting a referral, the increase in $\tilde{v}_1$ associated with an increase in $\lambda$ implies that a larger decrease in $\frac{R_b}{n_b}$ than in $\frac{R_g}{n_g}$. As $\frac{R_b}{n_b}>\frac{R_g}{n_g}$, it follows that the employment rate of green workers before the firing stage, given in \eqref{eq:emplgreen}, is increasing in $\lambda$. 

However, the employment rate of green workers, and thus of blue workers, at the end of the period may increase or decrease in $\lambda$. As noted in the discussing leading up the this proposition, there are two competing effects on the employment rate of green workers at the end of the period: in essence, the reliance on pool 1 increases, while firms are more reluctant to hire from pool 2. As the first effect occurs chronologically before the second, we are able to sign the change in green employment before the firing stage (see above). To sketch examples proving the claimed ambiguity, note that the former effect relies on there being a mass of workers (on the referral market) with values near $\tilde{v}_1$, whereas the latter effect relies on there being a mass of workers (in pool 1) with values near $\tilde{v}_2$. By appropriately placing mass points in the value distribution $F$, either effect can dominate. 

Finally, total production is increasing in $\lambda$ as $\tilde{v}_1$ is and by \eqref{eq:prodfir}; production before firing is decreasing in $\lambda$ as $\tilde{v}_1$ is and by Lemma~\ref{lemma:monotonicity}; and production after firing is increasing in $\lambda$ as $\tilde{v}_2$ is decreasing in $\lambda$ and by the usual accounting exercise.

If $\lambda'>\lambda=0$, it is clear from the proof so far that production is greater given $\lambda'$ than given $\lambda$ (also $\tilde{v}_1'>\tilde{v}_1$). Thus, we know that the employment rate of green workers at before the firing stage is greater given $\lambda'$ than given $\lambda$ which equals the employment rate of green workers at the end of the period as no firing takes place. The firing stage will only further increase the employment rate of green workers. In particular, as $\tilde{v}_1>\tilde{v}>\tilde{v}_2$, the masses of fired workers by group are proportional to the population sizes whereas the additional draw from pool 2 is still more likely to be a green worker.\end{proof}

\begin{proposition}
\label{prop:macro}
Suppose a mass $\kappa<1$ of firms no longer hire on the job market. Then:
\begin{itemize}
    \item total production decreases;
    \item production per employed worker increases; and
    \item the old wage distribution of wages first-order stochastically dominates the new one.
\end{itemize}
All comparisons are strict if $\kappa>0$.
\end{proposition}

\begin{proof}[Proof of Proposition~\ref{prop:macro}]
Let $P(0)$ and $P'(0)$ denote the probabilities of not getting a referral before and after a mass $\kappa$ of firms shut down.

We first show that total production decreases. Consider the new equilibrium, after some firms shut down. Now allow the firms that shut down to hire from the pool; this can only increase production. But then we are comparing two economies with a unit mass of firms but with different probabilities of not getting a referral. As $P(0)\leq P'(0)$, with $P(0)>P'(0)$ when $\kappa>0$, the conclusion follows from Proposition~\ref{lemma:inequality_and_productivity}.

Next, consider the average production of the remaining firms. The discussion  preceding the lemma shows that at the initial hiring threshold, the value in the pool exceeds the initial value in the pool; and strictly so if $\kappa>0$. Due to the uniqueness of the equilibrium threshold (Lemma~\ref{eqw}), it follows that the equilibrium threshold increased. After some firms shut down, a mass $(1-P'(0))n$ of firms will have on average the max of a random worker's value and the equilibrium threshold as their production; a complimentary mass, $1-\kappa -(1-P'(0))$, will have on average the equilibrium threshold as their production. The claim follows as $(1-P(0))n-(1-P'(0))n\leq \kappa$ and since the equilibrium threshold increases.

Lastly, recall that by Lemma~\ref{eqw}, wage of worker $i$ is $v_i-\tilde{v}+\ubar{w}$ if $i$ has more than one referral and $\ubar{w}$ otherwise, where $\tilde{v}$ is the equilibrium threshold. As the equilibrium threshold increases and since there are fewer workers with multiple referrals (both strictly so if $\kappa>0$), the wage distribution after firms shut down is first-order stochastically dominated by the initial wage distribution.
\end{proof}

\subsection{Omitted Examples}
\label{online:examples}

\begin{example}
\label{ex:cycle_pool}
Consider $h_b=h_g\geq 1/2$ and $n_b=n_g>1/2$ and suppose firms do not hire from the pool when initial employment is $e_b=e_g=1/2$. Then, only workers from the referral market are hired and the employment masses, while equal, are less than $1/2$. This reduces the number of referrals and thus increases the probability of not getting a referral which drives up the expected value in the pool. As a result, in the next period, firms may hire from the pool resulting in a unit-mass of employed workers equal by group; i.e., a cycle.
\end{example}

\begin{example}
\label{ex:cycle_bias}
Let $h_b=1;h_g=.5$ and $n_b=n_g>1/2$ so that referrals lean towards blues. Let $F$ be a three-value distribution with $v_H>v_M>v_L$ and small masses at the extremes. Start with equal employment rates and pick values $v_H,v_L$ such that $v_M$ is slightly above the hiring threshold. Blue employment increases as blues get more referrals than greens. This leads to more concentration, an increase in $P(0)$ and hence an increase in the hiring threshold. so that the hiring threshold is now above $v_M$. Then only most hiring comes from the pool and employment rates are mostly equal, and then the cycle repeats.
\end{example}

\begin{example}
\label{ex:algfairness}
Let $F$ be a three-value distribution with  $v_H>v_M>v_L$ and probabilities $(\frac{1-\epsilon}{2},\frac{1-\epsilon}{2},\epsilon)$, respectively, where $0<\epsilon<1$. The productivity values are high enough so that there will always be hiring from the pool. Let $n_G=n_B=1-\delta$, where $0<\delta<\frac{1}{2}$, and assume all blues and $\delta$ greens receive a referral. 

We begin with the regime where firms observe the group identity of workers in the pool. For every $\epsilon$, $v_H,v_M$, and $v_L$ can be chosen such that the hiring threshold on the referral market is exactly at  $v_M$. Then $(1-\epsilon)(1-\delta)$ blues are hired on the referral market. Green employment is thus at most $1-(1-\epsilon)(1-\delta)$  (which is arbitrarily close to zero for small $\epsilon$ and $\delta$). 

Let us consider the regime where the referral status is observable. Then the hiring threshold strictly increases: given the initial threshold, hiring from the pool becomes strictly more attractive because unreferred workers are targeted instead of the (slightly) tainted green workers. Hence, workers with value $v_M$ will not be hired on the referral market, implying that only $\frac{1-\epsilon}{2}(1-\delta)$ blues are hired via referrals. Since green workers are still hired more from the pool than blues, at most $(1-\frac{1-\epsilon}{2})/2$ of blue workers are hired from the pool. We can thus bound blue employment from above by summing the last two expressions; i.e., by $\frac{1-\epsilon}{2}(1-\delta)+ (1-\frac{1-\epsilon}{2})/2$ (which is arbitrarily close to $\frac{3}{4}$ for small $\epsilon$ and $\delta$). Finally, note that for small enough $\epsilon$ and $\delta$, green employment must have increased.
\end{example}

\subsection{Correlated Values}\label{corrvalues}
Homophily may not only occur along group dimensions, so that blues tend to refer blues and greens tend to refer greens, but homophily could also occur along productivity dimensions; i.e., so that workers tend to refer workers who have a similar value to their own.

We study the effects of such homophily by means of an example. We consider a two-value value distribution with equal amounts of high- and low-value workers: Let $F$ be given by
\begin{equation}
v_i = \begin{cases}
v_H \quad &\text{wp } 1/2 \\
v_L \quad &\text{wp } 1/2,
\end{cases}
\end{equation}
with $v_H>v_L$.

Let $\alpha\geq  1/2$ denote the ``inbreeding-bias'': the probability that the connection of some worker is of the same productivity type as that worker. For $\alpha=1/2$, then values are uncorrelated and so there is no imbreeding-bias.

For simplicity, let $n=2$; i.e., $n_H=n_L=1$, where $n_H,n_L$ are the masses of high- and low-value workers, respectively; and,  we work with a referral distribution that has weight only on $P(0)$ and $P(1)$.

On the referral market, firms will always hire a high-value worker and reject a low-value worker. Denoting the employment of high-value workers by $e_H$, the probability that a high-value worker gets a referral is 
$$
    P_H(1|e_H,1-e_H)\coloneqq e_H\alpha +(1-e_H)(1-\alpha).
$$
The steady state employment rate (or mass) of high-value workers, $e_H$, solves
\begin{equation*}
    P_H(1|e_H,1-e_H)+(1-P_H(1|e_H,1-e_H))\frac{1-P_H(1|e_H,1-e_H)}{n-P_H(1|e_H,1-e_H)} = e_H.
\end{equation*}
Solving the above expression for $e_H$ yields 
\begin{equation*}
e_H = \frac{1+\alpha-\sqrt{(1-\alpha)(5-\alpha)}}{2(2\alpha-1)}.
\end{equation*}
This is increasing in $\alpha$ as
\begin{equation*}
\frac{\partial }{\partial \alpha}e_H = \frac{7-5a-3\sqrt{(1-a)(5-a)}}{2(2a-1)^2 \sqrt{(1-a)(5-a)}},
\end{equation*}
which is positive since $5\alpha + 3 \sqrt{(1-\alpha)(5-\alpha)}$ is maximized for $\alpha=1/2$, for which it equals $7$. For $\alpha>1/2$, the above expression is thus strictly positive.

We depict the mass of employed high-value workers in steady state as a function of the degree of value-homophily in Figure \ref{fig:correlated_values}.
\begin{figure}[!ht]
\centering
\includegraphics[]{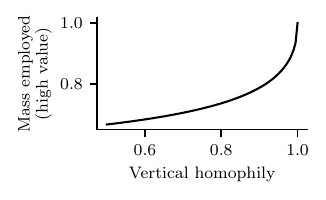}
\caption{The parameter values are: $n_H=n_L=1$ so that $n=2$.
We plot the mass of high-value workers employed in steady state as a function of the degree of value-homophily, $\alpha$.}
\label{fig:correlated_values}
\end{figure}
We remark that value-homophily increases productivity, since in this example the mass of high valued workers translates directly into the overall productivity.

\subsection{Constrained Efficiency}\label{sec:constrained_efficiency}
It is clear that it is without loss to constrain the social planner to use a threshold hiring strategy to maximize total production, the sum of the values of employed workers and the outside options of unemployed workers. Consider a social planner who freely chooses the threshold $\tilde{v}$, instructs a fraction $r$ of firms with a referred worker with value exactly equal to $\tilde{v}$ to hire that worker, and decides whether firms hire from the pool.

\begin{lemma}\label{lemma:constrained_efficiency}
    Suppose the outside option of workers is equal to $\ubar{w}$. Then the social planner maximizes total production by choosing $\tilde{v}$ equal to the unique equilibrium threshold given in Lemma \ref{eqw} and $r$ arbitrarily. 
\end{lemma}

\begin{proof}[Proof of Lemma~\ref{lemma:constrained_efficiency}]
    Note that the choice of firms to hire from the pool or not is efficient in equilibrium. Suppose that it is productivity maximizing for the social planner to choose a hiring threshold so that firms optimally do not hire from the pool. Clearly, this threshold to maximize productivity is $\ubar{w}$; and further this is only productivity maximizing if $E_{\tilde{v},r}E_{v_i}|i \in \pool] \leq \ubar{w}$. But then $\tilde{v},r$ would solve \eqref{eqw_con} and so constitute the unique equilibrium threshold.
    
    Now suppose that it is productivity maximizing for the social planner to choose a hiring threshold so that firms optimally hire from the pool. Clearly, the social planner chooses $\tilde{v}\geq \ubar{w}$ and hiring from the pool is only optimal if further $E_{\tilde{v},r}E_{v_i}|i \in \pool] \geq \ubar{w}$. In this case, as 
    $$
        E[v_i]=\frac{1}{n}E_{\tilde{v},r}[v_i|i \text{ is employed}]+\frac{n-1}{n}E_{\tilde{v},r}[v_i|i\text{ is unememployed}]
    $$
    and $E_{\tilde{v},r}[v_i|i\text{ is unememployed}]=E_{\tilde{v},r}[v_i|i\in \pool]$, total production is maximized when the expected value in the pool is minimized. But when the  this is exactly the equilibrium condition in \eqref{eqw_con} when in equilibrium firms hire from the pool: Referred workers with values below the expected value in the pool are entering the pool, while those with values above are hired, thus minimizing the expected value in the pool. 
\end{proof}

\subsection{Costly skill investment}\label{sec:online_skill}
\begin{proposition}\label{prop:inv}
For every $v_L,v_H,\rho$, with $\rho<1$, 
\begin{itemize}
    \item if $c$ is small enough, then, in the unique steady state, every worker invests;
    \item if $c$ is large enough, then, in the unique steady state, no worker invests; and
    \item there exists $\ubar{c},\bar{c}$, with $\ubar{c}<\bar{c}$, such that for all $c \in (\ubar{c},\bar{c})$, there is a steady state with employment bias and in which only blue workers invest.
\end{itemize}
\end{proposition}

\subsection{Allocating Firms' Profits in the Gini Comparative Statics}\label{online-appendix:results}

Our analysis of the Gini coefficient  focused on wages. Firms in this economy are earning profits since they can hire workers at $\ubar{w}$ and earn a higher expected value. Accounting for who gets those profits as income can affect the inequality calculations.

Firms earn the same profit from workers they compete (those with high values and $P(2+)$) as hiring a worker from the pool. This means that expected profits take a simple form (here presuming that the expected value from the pool exceeds the minimum wage): Profits are:
\begin{equation}
\label{profits}
	\Pi =   nP(1)f_H v_H + (1- P(1)f_H){\tilde{v}} -\ubar{w}.
\end{equation}
So, firms' profits depend entirely on how many people get just one referral, as well as what the threshold is and hence the lemons effect. Thus, firms' profits can be fully characterized if we know $P(0)$ and $P(2+)$ (and hence $P(1)$).

To provide some simple intuition, let us consider a distribution for which workers get either $0,1$ or $k$ referrals. Together with the assumption that the total number of referrals is constant equaling $1$, this implies that $P$ is completely characterized by one parameter: $P(k)$. In that case, simple but tedious calculations show that $\frac {\partial \Pi}{\partial P(k)}=- \frac{f_Lf_H(v_H-v_L)(f_L+(n-1)k)}{(P(0)+(1-P(0))f_L)^2}<0$. This is not obvious, since once again the lemons effect has a counteracting force, but in this situation the derivative can be unambiguously signed.\footnote{The weaker assumption that an increase in the probability of getting multiple referrals implies an increase in the probability of not getting a referral is not sufficient to sign the effect on profit.}

This means that if profits are distributed uniformly across all workers, then concentrating referrals (an increase in $P(k)$ which now corresponds to both an increase in the probability of not getting a referral and multiple referrals) decreases profits, and so decreases income uniformly across all workers which increases the Gini coefficient and thus inequality.\footnote{An easy way to see this is to note that the Gini decreases in $W_L$ (see equation (\ref{giniexp})) and so lowering the level of all incomes decreases $W_L$, thereby increasing the Gini.} Thus, compared to Proposition~\ref{GiniChange}, this is another pressure increasing inequality. If profits instead go to some special class of citizens who are owners of the firms, then it depends on who they are, and so then the further effects are ambiguous.

\end{document}